\documentclass[12pt, a4paper, fleqn]{article}
\usepackage{amsmath}
\usepackage{amssymb}
\usepackage{graphicx}
\usepackage{caption}
\usepackage{adjustbox}
\usepackage{fancyhdr}
\usepackage{lmodern}
\usepackage{placeins}
\usepackage{appendix}
\usepackage{bbm}
\usepackage{color}
\usepackage[normalem]{ulem}

\usepackage[left=2cm,right=2cm,top=3cm,bottom=2.5cm]{geometry}
\usepackage[utf8]{inputenc}

\usepackage{amsthm}
\newtheorem{theorem}{Theorem}
\newtheorem{lemma}[theorem]{Lemma}
\newtheorem{definition}{Definition}

\title{\vspace{-1cm} Stochastic Model Predictive Control for Networked Systems with Random Delays and Packet Losses in All Channels}

\renewcommand\footnotemark{}
\author{Marijan Palmisano $(1)$, Martin Steinberger $(1)$, Martin Horn $(1,2)$
	\thanks{\small $(1)$ Marijan Palmisano, Martin Steinberger and Martin Horn are with the Institute of Automation and Control, Graz University of Technology, Graz, Austria, (e-mail: marijan.palmisano@student.tugraz.at).}%
}
\date{}

\begin{document}
	\maketitle

\begin{abstract}
A stochastic Model Predictive Control strategy for control systems with communication networks between the sensor node and the controller and between the controller and the actuator node is proposed. Data packets are subject to random delays and packet loss is possible; acknowledgments for received packets are not provided.

The expected value of a quadratic cost is minimized subject to linear constraints; the set of all initial states for which the resulting optimization problem is guaranteed to be feasible is provided. The state vector of the controlled plant is shown to converge to zero with probability one.
\end{abstract}

\section{Introduction}
A typical setup considered in Model Predictive Controller (MPC) design consists of a discrete-time plant and a discrete-time controller where the output of the controller at time step $k$ is the actuating variable of the plant $u_k$ and the input of the controller at time step $k$ is the state variable of the plant $x_k$. The controller output is determined by minimizing a cost function which depends on current and predicted values of state and actuating variable. For details, see, e.g., \cite{Mayne00}.

In a networked control system (NCS), controller and plant are connected via communication networks. Data packets are subject to random delays and packet losses which are two of the major challenges in networked control; see \cite{ZHANG2017} or \cite{Xia2015}. 

While MPCs are also developed for NCS where no random delays can occur between controller and actuator node like in \cite{Bahraini2021}, \cite{Munoz2008}, \cite{Trodden2023} or \cite{Loma-Marconi2024}, such delays appear to be particularly challenging. There are numerous works on designing MPCs for NCS with such delays. For example, delayed packets are discarded in \cite{Quevedo2011} which effectively turns random delays into packet losses.

In \cite{Varutti2009}, buffers are used to compensate for the randomness of the delays by ensuring that the evolution of the actuating variable applied to the plant is equivalent to the evolution of the actuating variable used in the prediction. Effectively, this replaces the random delay with the maximal delay.

Another possible approach, e.g. presented in \cite{Reble2013}, \cite{Quevedo2010} and \cite{Quevedo2007}, involves predicting future values of state and actuating variable using the nominal plant model, i.e. without including the random effects. This introduces a mismatch between the actual setup and the prediction model which can be expected to deteriorate the performance of the closed control loop. 

An approach which avoids such a mismatch is to include the random effects in the prediction model and to minimize the expected value of the cost. It is typically assumed that information about data received by the plant is instantly available to the controller; e.g. $(a)$ via 'error free receipt acknowledgments' in \cite{Mishra2020}, $(b)$ via 'TCP-like protocols' in \cite{Li2018} or $(c)$ implicitly by using such information in the optimization problem to be solved by the MPC in \cite{Wu2009}  and \cite{Liu2011}. 

While minimizing the expected value of the cost function results in the best expected performance, it is not always possible to provide 'error free receipt acknowledgments'. In \cite{Rosenthal2019}, an MPC is presented which minimizes the expected value of the cost function while taking into account that the acknowledgments might be subject to random delays and packet losses. However, this MPC is designed for plants without constraints on the states and inputs; this also applies to the MPCs presented in \cite{Zou2010}. A different strategy is presented in \cite{Mao2020} where the MPC provides 'correction terms' to a local controller which is connected directly to the plant.

Another approach which avoids a mismatch between the actual setup and the prediction model is minimizing the worst-case cost instead of the expected value of the cost like in \cite{Pan2023} where the robust MPC design method presented in \cite{Kothare1996} is extended to NCS.

\subsubsection*{Main Contribution}
To our knowledge, no Model Predictive Control Strategy $(i)$ minimizing the expected value of a cost function has yet been developed $(ii)$ for NCS with random delays and packet losses between controller  and actuator node where $(iii)$ acknowledgments are also subject to random delays and packet losses and $(iv)$ constraints on the states and inputs have to be satisfied. In this work, one such control strategy is presented for NCS with a linear time-invariant (LTI) plant for a quadratic cost and linear inequality constraints.

Consider the common control problem of designing an MPC for an LTI plant which is connected directly to the controller in order to minimize a quadratic cost while satisfying linear equality constraints. The only difference between this control problem and the control problem considered in this work is that plant and controller are separated, i.e. they can only communicate via communication networks. 

The assumptions on the properties of these networks are quite common as well: (i) there are no instantly available acknowledgments about received packets, i.e. 'UDP-like protocols' are used, (ii) packet are timestamped and (iii) random delays and/or packet loss can occur where (iv) delays and number of consecutive packet losses are bounded. The the considered random network effects are handled by minimizing the expected value of the cost in order to achieve the best possible expected performance.

\subsubsection*{Notation}
$\succeq,\succ$ denote (semi-) definiteness; $\geq,>$ denote element-wise inequalities. $\otimes$ denotes the Kronecker product. $I_n$ is the $n\times n$ identity matrix, $1_{n\times m}$ is a $n\times m$ matrix where all elements are $1$, $0_{n\times m}$ is a $n\times m$ matrix where all elements are $0$, $1_n = 1_{n\times 1}$ and $0_{n} = 0_{n\times 1}$. $(a_k)_{i\leq k\leq j}$ with $i,j\in\mathbb{Z}$ denotes the sequence $(a_i, a_{i+1}, ..., a_j)$. $\frown$ denotes appending an element to a sequence, e.g. $(a_1,a_2)^\frown a_3 = (a_1,a_2,a_3)$. $\mathbb{P}[A]$ denotes the probability of event $A$ and $\mathbb{P}[A|B]$ denotes the conditional probability of event $A$ given event $B$.

\section{ Considered Feedback Loop}\label{sec_feedback_loop}
The considered networked control systems consist of one plant and one controller; data packets transmitted from plant to controller and vice versa are subject to random delays and packet loss is possible. Delays and number of consecutive packet losses are bounded and packets are timestamped. Information about packets received at the actuator node is transmitted to the controller using a network connection and is therefore subject to random delays and packet loss is possible. 

This setup is depicted in Fig.~\ref{fig1}, where the plant with the state $x_{k} \in \mathbb{R}^n$ and the actuating variable $u_k\in\mathbb{R}^m$ is a discrete-time LTI system
\begin{align}
	x_{k+1} &= Ax_k+Bu_k, \quad k\in\mathbb{N}_0\label{eq_LTI}
\end{align}
with the initial state $x_0$, where $(A,B)$ is stabilizable. State and actuating variable have to satisfy linear constraints
\begin{subequations}\label{constraints}
	\begin{align}
		M_{x}x_k &\leq n_x \quad \forall k\in\mathbb{N}_0,\quad M_{x}\in\mathbb{R}^{a_x \times n},\ n_x\in\mathbb{R}^{a_x} \\
		M_{u}u_k &\leq n_u \quad \forall k\in\mathbb{N}_0,\quad M_{u}\in\mathbb{R}^{a_u \times m},\ n_u\in\mathbb{R}^{a_u}
	\end{align}
\end{subequations}
where $\left\{x_{k}\in\mathbb{R}^n|M_{x}x_{k} \leq n_x\right\}$ and $\left\{u_{k}\in\mathbb{R}^m|M_{u}u_{k} \leq n_u\right\}$ are closed sets containing the origin as an inner point.
\begin{figure}[t]
	\centerline{\includegraphics[width=0.8\columnwidth]{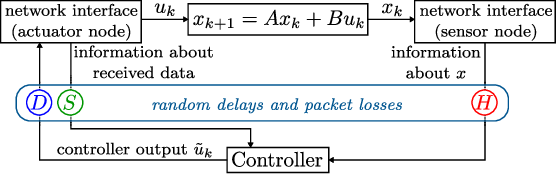}}
	\caption{Considered class of networked control systems; random network effects are represented by the stochastic processes $\mathbf{D}$, $\mathbf{H}$ and $\mathbf{S}$.}
	\label{fig1}
\end{figure}

\subsection{Network Model}
The network used to transmit data from the controller to the actuator node is represented by a discrete-time stochastic process $\mathbf{D} = \left(D_{k}\right)_{k\in\mathbb{N}_0}$ where $D_k$ is the age of the most recently transmitted available packet like in \cite{9115012} for $D_k \leq k$. The first packet is transmitted at time step $k = 0$ and $D_k > k$ means that no packet has been received until time step $k$. The random variable $D_k$ can only increase by $1$ each time step and is bounded by $\underline{d} \leq D_k\leq \overline{d}$ with $\underline{d},\overline{d}\in \mathbb{N}_0$; for $\delta\in\mathbb{N}$, this can be written as
\begin{align}
	&
	\mathbb{P}\left[D_{k+1}\!=\! \delta|D_k\! = \! d_k\right] > 0 \, \Rightarrow\, \underline{d}\leq \delta\leq\min\!\left(\overline{d},d_k \! +\! 1\right).
	\label{eq_non-zero_probabilities}
\end{align}

The other networks are modeled analogously: the network used to transmit data from the sensor node to the controller is represented by $\mathbf{H} = \left(H_{k}\right)_{k\in\mathbb{N}_0}$ with $\underline{h} \leq H_{k} \leq \overline{h}$; the network used to transmit data from the actuator node to the controller is represented by $\mathbf{S}=\left(S_{k}\right)_{k\in\mathbb{N}_0}$ with $\underline{s} \leq S_{k} \leq \overline{s}$.

While there are no further assumptions on $\mathbf{H}$ and $\mathbf{S}$, $\mathbf{D}$ is assumed to be a homogeneous Markov Process, i.e.
\begin{align}
	&\mathbb{P}\big[D_{k+1}\!=\! d_{k+1} \big| (D_{\tilde{k}})_{0 \leq\tilde{k}\leq k} \!=\! (d_{\tilde{k}})_{0 \leq\tilde{k}\leq k}\big]
	=\mathbb{P}\left[D_{k+1}\!=\! d_{k+1}|D_k\! = \! d_k\right] \ \forall k \in \mathbb{N}_0,
	\nonumber\\
	&\mathbb{P}\left[D_{k+1}\!=\! \beta|D_k\! = \! \alpha\right] = \mathbb{P}\left[D_{\tilde{k}+1}\!=\! \beta|D_{\tilde{k}}\! = \! \alpha\right]  \ \forall k,\tilde{k} \in \mathbb{N}_0\label{eq_Markov}.
\end{align}
This Markov Process is described by the known initial distribution $\mu$ and the known transition probabilities $\Phi$ given by
\begin{subequations}
	\begin{align}
		\mu(d_0) &= \mathbb{P}\left[D_{0}=d_{0}\right] > 0\quad \forall d_0\in\left[\underline{d},\overline{d}\right]\label{eq_mu}\\
		\Phi(d_{k},d_{k+1}) &= \mathbb{P}\left[D_{k+1}=d_{k+1}|D_k=d_k\right], \ k\in\mathbb{N}_0; \label{eq_Phi}
	\end{align}
\end{subequations}
the resulting $n$-step transition probabilities $\Phi_n$ are denoted as
\begin{align}
	\Phi_n(d_{k},d_{k+n}) &= \mathbb{P}\left[D_{k+n}=d_{k+n}|D_k=d_k\right],\ k,n\! \in\! \mathbb{N}_0. \label{eq_Phi_n}
\end{align}
In order to avoid case distinctions, it is assumed that $\overline{d}>\underline{d}$, i.e. that controller outputs are actually subject to random effects.

\subsection{Transmitted Data Packets}
\subsubsection{Sensor Node}
Let $K_h$ be the first time step at which a packet from the sensor node is received at the controller. Once $k \geq K_h$, the controller has access to $K_h$ and $(H_{\tilde{k}})_{K_h \leq \tilde{k}\leq k}$. By transmitting $(x_{\tilde{k}})_{\min(0, k-\overline{h}+\underline{h}) \leq \tilde{k} \leq k}$ at each time step $k \geq 0$, it is ensured that the controller has access to the entire sequence $(x_{\tilde{k}})_{0 \leq \tilde{k}\leq k-H_k}$ at each time step $k\geq K_h$.

\subsubsection{Controller}
The controller transmits the prediction
\begin{align} \label{eq_uk_tilde_general}
	\tilde{u}_k &= \big[\begin{matrix}
		\tilde{u}_k^{(\overline{d})^T}\!\!&
		\tilde{u}_k^{(\overline{d}-1)^T}& \hdots& \tilde{u}_k^{(\underline{d})^T}
	\end{matrix}\big]^T
	\in\mathbb{R}^{\tilde{m}}
\end{align}
where $\tilde{m} = (\overline{d}-\underline{d}+1)m$ to the actuator node at all time steps $k \geq 0$. Let $K_d$ be the first time step at which a data packet is received at the actuator node. Once $k \geq K_d$, $u_k$ is determined via
\begin{align}
	u_{k} = \tilde{u}_{k-D_k}^{(D_k)}. \label{eq_uk}
\end{align}
from the most recent available packet. While $k < K_d$, $u_k = 0$ is applied which can also be written as \eqref{eq_uk} by defining $\tilde{u}_k = 0\ \forall\, k < 0$.

Due to \eqref{eq_uk}, $\tilde{u}_k^{(d)}$ in \eqref{eq_uk_tilde_general} is the value applied as $u_{k+d}$ if $\tilde{u}_k$ is the most recent available controller output at time step $k+d$.

\subsubsection{Actuator Node} Once $k \geq K_d$, the actuator node transmits $\left(D_{\tilde{k}}\right)_{\max(K_{d},k-\overline{s}+\underline{s})\leq \tilde{k} \leq k}$; empty packets are transmitted while $k < K_d$. Let $K_s$ be the first time step at which a packet from the actuator node is received at the controller and let $K_{sd}\geq K_s$ be the first time step at which a non-empty data packet is received. Once $k \geq K_s$, the controller has access to $K_s$ and $(S_{\tilde{k}})_{K_s \leq \tilde{k}\leq k}$. Once $k \geq K_{sd}$, the controller has access to $K_{sd}$, $K_d$ and the entire sequence $\left(D_{\tilde{k}}\right)_{K_{d}\leq \tilde{k} \leq k-S_k}$. While $K_{s} \leq k <K_{sd}$, the controller has the information that $K_d > k-S_k$. Note that this implies that $\left(u_{\tilde{k}}\right)_{0\leq \tilde{k} \leq k-S_k}$ can be determined by the controller for all $k \geq K_s$.

\subsubsection{Information Set}
All data available to the controller at time step $k$ assuming that all data is maintained is referred to as information set $\mathcal{I}_k$. This implies $\mathcal{I}_{k}\subseteq \mathcal{I}_{k+i}\, \forall i\geq0$. The amount of storage required for maintaining this set is not bounded so it is not actually maintained by the controller. While the controller is designed assuming that $\mathcal{I}_k$ is available at time step $k$, maintaining a finite subset of $\mathcal{I}_k$ is sufficient for implementing the resulting control law.

\section{Model Predictive Controller}
Consider an MPC where an optimization problem
\begin{subequations}\label{eq_optimization_problem}
	\begin{align}
		&\underset{(\hat{u}_{k}^{(i)})_{\underline{d}\leq i\leq N-1}}{min} \!\!\!\! \mathbb{E}\left\{\left. \sum_{i=0}^{\infty} \Big(x_{k+i}^TQx_{k+i} + u_{k+i}^TRu_{k+i}\Big) \right|\mathcal{I}_k\right\}\label{eq_cost_stochastic}
		\\ 
		&s.t.\ u_{k+i} = \begin{cases}
			\tilde{u}_{k+i-D_{k+i}}^{(D_{k+i})}& i < D_{k+i}\\
			\hat{u}_{k}^{(i)}& D_{k+i} \leq i \leq N-1\\
			\kappa_{k}^{(i)}& i\geq N
		\end{cases}\label{eq_uki_stochastic}\\
		&\quad\ \ \ M_{u}u_{k+i} \leq n_u, \, M_{x}x_{k+i} \leq n_x \quad \forall i\in [0,N-1]\label{constr_stochastic}\\
		&\quad\ \ \ x_{k+N} \in \mathcal{X}_N(\mathcal{I}_k)\label{terminal_stochastic}
	\end{align}
\end{subequations}
with the finite optimization horizon $N\in\mathbb{N}$ is solved at each time step $k\geq K_h$ where $R = R^T \succ 0$ and $Q$ can be written as $Q=C^TC \succeq 0$ with some $C$ such that $(A,C)$ is detectable. 

Once $k\geq K_h$, the controller output \eqref{eq_uk_tilde_general} is chosen via $\tilde{u}_{k}^{(d)} = \hat{u}_{k}^{(d)}$ which requires that $N\geq \overline{d}+1$. While $k < K_h$, i.e. while the controller has not received any information about the state vector, the controller output is set to zero, i.e.
\begin{align}
	&\tilde{u}_k
	= \begin{cases}
		\begin{bmatrix}
			\hat{u}_{k}^{{(\overline{d})}^T}&\hat{u}_{k}^{{(\overline{d}-1)}^T}&  \hdots &  \hat{u}_{k}^{{(\underline{d})}^T}
		\end{bmatrix}& k \geq K_h\\
		0 & k < K_h.
	\end{cases}
	\label{eq_uk_tilde}\\[-7mm]\nonumber
\end{align}

In \eqref{eq_optimization_problem}, the conditional expected value of a cost given all information available at time step $k$ is minimized in order to achieve the best possible expected performance. The constraint \eqref{eq_uki_stochastic} takes into account that $(a)$ $u_{k+i}$ is chosen from previous controller outputs while no packet transmitted at time step $k$ or later is received at the actuator node, $(b)$ the optimization variables are the actuating variables while $D_{k+i} \leq i \leq N-1$ and $(c)$ a control law $u_{k+i} = \kappa_{k}^{(i)}$ has to be chosen for $i \geq N$ in order to evaluate the cost function. 

\begin{theorem}\label{th_1}
	Consider the feedback loop presented in Section~\ref{sec_feedback_loop}. For any $N \geq \overline{d}+1$, a stabilizing control law $u_{k+i} = \kappa_k^{(i)}\ \forall\, i \geq N$, a terminal constraint set $\mathcal{X}_N(\mathcal{I}_k)$ and a non-empty set $\mathcal{X}_0 = \left\{ x_0\in \mathbb{R}^n \big|M_{0}x_0\leq n_{0}\right\}$ exist such that choosing the controller output via \eqref{eq_uk_tilde} for all $k\geq 0$ where $(\hat{u}_{k}^{(i)})_{\underline{d}\leq i\leq \overline{d}}$ is determined by solving \eqref{eq_optimization_problem} for $k \geq K_h$ yields the following properties if the initial state satisfies $x_0 \in \mathcal{X}_0$:
	\begin{enumerate}
		\item[a)] \eqref{eq_optimization_problem} is feasible at all time steps $k \geq K_h$.
		\item[b)] The constraints \eqref{constraints} are guaranteed to be satisfied.
		\item[c)] The state vector converges to zero with probability one, i.e. \begin{align*}\mathbb{P}\big[\lim\nolimits_{k\to\infty}x_k = 0\big] = 1.\end{align*}
		\item[d)] There exists an optimization problem that $(i)$ is equivalent to \eqref{eq_optimization_problem} in the sense that the resulting optimal value for $\hat{u}_k = \begin{bmatrix}
			\hat{u}_{k}^{(N-1)^T}& \hdots& \hat{u}_{k}^{(\underline{d})^T}
		\end{bmatrix}^T$ is the same which $(ii)$ can be written in the form
		\begin{align}
			&\underset{\hat{u}_{k}}{min} \left(\hat{u}_k^TV_{k}\hat{u}_k + \hat{u}_k^Tv_{k}\right)\quad s.t. \quad W_{k}\hat{u}_k \leq w_{k} \label{eq_final_opt_pr_form}
		\end{align}
		where $V_{k}$, $v_{k}$, $W_{k}$ and $w_{k}$ can be computed from $(a)$ values that can be computed offline in advance and $(b)$ probabilities and an extended state vector.
	\end{enumerate}
\end{theorem}

\section{Proof for Theorem~\ref{th_1}}\label{sec_proof}
\subsection{Stabilizing Control Law and Terminal Constraint Set}
\subsubsection{Stabilizing Control Law}\label{sec_stabilizing}
Consider a controller with the input $x_k$ and the output $u_k$ and consider the cost function $J = \sum_{k = 0}^{\infty}\left(x_k^TQx_k + u_k^TRu_k\right)$. Minimizing $J$ with respect to $(u_k)_{k\geq 0}$ in the unconstrained case yields the well-known Linear Quadratic Regulator (LQR), i.e. $u_k = -L x_k$ with $L = (R + B^TPB)^{-1}B^TPA$ resulting in $J = x_0^TPx_0$ where $P\succ 0$ is the solution of the Discrete-Time Algebraic Riccati Equation $P = Q + A^TPA - A^TPB \left(R + B^TPB\right)^{-1} B^TPA$.

The stabilizing control law $\kappa_k^{(i)}$ should be chosen such that $u_{k+i} = \kappa_k^{(i)}\ \forall\, i\geq N$ can actually be achieved by choosing the controller outputs $(\tilde{u}_{k+i})_{i\geq 0}$ accordingly. Otherwise, a mismatch between the actual setup and the prediction with \eqref{eq_uki_stochastic} is introduced. This issue occurs for $\kappa_{k}^{(i)} = -Lx_{k+i}$.

When choosing $\kappa_{k}^{(i)} = -L\mathbb{E}\left\{\left.x_{k+i}\right|\mathcal{I}_k\right\}$ instead, one sequence $(u_{k+i})_{i\geq N}$ has to stabilize the plant for all possible $x_{k+N}$ resulting in unnecessarily conservative constraints. 
\begin{definition}[Proposed Stabilizing Control Law]\label{def_stabilizing}
	To avoid both of the above issues, the following choice is proposed:
	\begin{subequations}\label{stabilizing_control_law}
		\begin{align}
			&\kappa_{k}^{(i)} = -Lx_{k+i} + L\!\!\! \sum_{j = -\hat{H}_k^{(i)}}^{\overline{d}-1}\!\! A^{i-1-j}B\Big(u_{k+j} -\mathbb{E}\left\{\left.u_{k+j}\right|\mathcal{I}_{k}\right\}\Big)\label{eq_kappa_ki_proposed}
			\\[-4.5mm]
			&\tilde{H}_k = \begin{cases}
				H_k& k < K_s\\
				\min(H_k,\! S_k-1\!)& k \geq K_s
			\end{cases}
			\label{eq_hk_tilde}
			\\
			& \hat{H}_k^{(i)} = \min(\overline{h}+N-i-1, \overline{s}+N-i-2,\tilde{H}_k).\label{eq_hk_hat}
		\end{align}
	\end{subequations}
\end{definition}
As shown in Appendix~\ref{Appendix_compute_x_Hk_hat}, $(\tilde{u}_{k+i})_{i\geq 1}$ can be chosen, i.e. $\tilde{u}_{k+i}$ can be computed from $\mathcal{I}_{k+i}$ for all $i \geq 1$, such that \eqref{eq_uki_stochastic} is actually applied to the plant for \eqref{eq_kappa_ki_proposed}. Since
\begin{align}
	\kappa_{k}^{(i)} &= -Lx_{k+i} \ \ \forall i \geq \hat{N}
	= \overline{d} + N-1+\min(\overline{h},\overline{s}\! -\! 1) \label{kappa_ki_i_geq_hat_N}
\end{align}
since $-\hat{H}_k^{(i)} \geq \overline{d}$ for all $i \geq \hat{N}$ due to \eqref{eq_hk_hat},  $\lim\nolimits_{k\to\infty}x_k = 0$ certainly holds for any $x_{k+N}$ if $u_{k+i} = \kappa_{k}^{(i)}$ for all $i \geq N$.

\subsubsection{Terminal Constraint Set}\label{sec_terminal_constr}
If $u_{k+i} = -Lx_{k+i}\ \forall\, i\geq 0$, $M_xx_{k+i}\leq n_x$ and $M_uu_{k+i}\leq n_u$ hold for all $i \geq 0$ if
\begin{align}
	\begin{bmatrix}
		M_x^T&
		-(M_u L)^T
	\end{bmatrix}^T\left(A-BL\right)^i x_{k} \leq \begin{bmatrix}
		n_x^T& n_u^T
	\end{bmatrix}^T.\label{eq_1254}
\end{align}
Using Algorithm 3.2 in \cite{Gilbert91}, the set of all $x_k$ for which \eqref{eq_1254} holds can be written as $M_N x_{k} \leq n_N$.

The value of $u_{k+i}$ for $D_{k+i} = d_{k+i}$ for $i \leq N-1$ 
according to \eqref{eq_uki_stochastic} used for prediction at time step $k$ is denoted by
\begin{align}
	\bar{u}_{k}^{(i)}(d_{k+i}) &= \begin{cases}
		\tilde{u}_{k+i-d_{k+i}}^{(d_{k+i})}& i < d_{k+i}\\
		\hat{u}_{k}^{(i)}& d_{k+i} \leq i \leq N-1
	\end{cases}\label{uk_bar}
\end{align}
in order to distinguish it from the value of $u_{k+i}$ used for prediction at time step $k+1$ which is denoted by $\bar{u}_{k+1}^{(i-1)}(d_{k+i})$.

The sequence $(D_{k+i})_{-\tilde{H}_k \leq i \leq \overline{d}-1}$ is denoted by $\vec{D}_k$; the set of all possible sequences $\vec{D}_k$ given $\mathcal{I}_k$ is denoted as $\mathcal{D}_k$, i.e.
\begin{subequations} 
	\begin{align}
		&\vec{D}_k = (D_{k+i})_{-\tilde{H}_k \leq i \leq \overline{d}-1}\label{eq_Dk_vec} \\
		&\vec{d}_k \in \mathcal{D}_k\ \Leftrightarrow\ \mathbb{P}[ \vec{D}_k = \vec{d}_k | \mathcal{I}_k] > 0.\label{eq_Dk_vec_set}
	\end{align}
\end{subequations}

\begin{definition}[Proposed Terminal Constraint Set]\label{def_terminal}
	$\mathcal{X}_N(\mathcal{I}_k)$ is chosen as the set of all $x_{k+N}$ for which \begin{subequations}\label{eq_7}
		\begin{align}
			&M_{u}u_{k+i} \leq n_u, \ M_{x}x_{k+i} \leq n_x \quad \forall i\in [N,\hat{N}-1]\\
			&M_N x_{k+\hat{N}} \leq n_N\label{eq_7.2}
		\end{align}
	\end{subequations}
	hold if $u_{k+i}$ can be written in the form
	\begin{subequations}\label{eq_form_uki_feas}
		\begin{align}
			& u_{k+i} = \sum\nolimits_{\vec{d}_k \in \mathcal{D}_k} p_{k,i}(\vec{d}_k) \bar{\kappa}_{k}^{(i)}(\vec{d}_k),\\
			& \bar{\kappa}_{k}^{(i)}(\vec{d}_k) = -Lx_{k+i} + L \sum\nolimits_{j = -\hat{H}_k^{(i)}}^{\overline{d}-1}  A^{i-1-j}B\left(\bar{u}_{k}^{(j)}(D_{k+j}) -\bar{u}_{k}^{(j)}(d_{k+j})\right)\label{kappa_k_bar}
		\end{align}
		for all $i \in [N, \hat{N}-1]$ where the scalar factors $p_{k,i}(\vec{d}_k)$ satisfy
		\begin{align}
			& \sum\nolimits_{\vec{d}_k \in \mathcal{D}_k} p_{k,i}(\vec{d}_k) = 1, \quad
			p_{k,i}(\vec{d}_k) \geq 0 \ \ \forall\vec{d}_k \in \mathcal{D}_k.\label{eq_9}
		\end{align}
	\end{subequations}
\end{definition}

Due to \eqref{eq_kappa_ki_proposed}, $u_{k+i} =\kappa_k^{(i)}$ can be written in the form \eqref{eq_form_uki_feas} with $p_{k,i}(\vec{d}_k) = \mathbb{P}\big[ \vec{D}_k=\vec{d}_k \big| \mathcal{I}_k\big]$. Therefore, \eqref{terminal_stochastic} implies that \eqref{eq_7} holds for \eqref{eq_uki_stochastic}. Furthermore, \eqref{eq_7} and \eqref{kappa_ki_i_geq_hat_N} imply that $M_{u}u_{k+i} \leq n_u$ and $M_{x}x_{k+i} \leq n_x$ for all $i \geq N$ for \eqref{eq_uki_stochastic}.

In other words, $\mathcal{X}_N(\mathcal{I}_k)$ is chosen such that \eqref{terminal_stochastic} and \eqref{constr_stochastic} imply that $M_{u}u_{k+i} \leq n_u$ and $M_{x}x_{k+i} \leq n_x$ are guaranteed to hold for all $i \geq 0$ if \eqref{eq_uki_stochastic} holds.

\subsection{Feasibility and Admissible Initial States}\label{sec_feasibility}
\begin{lemma}[Recursive Feasibility]\label{le_recursive}
	If $(a)$ $\kappa_k^{(i)}$ and $\mathcal{X}_N(\mathcal{I}_k)$ are chosen according to Definition~\ref{def_stabilizing}~and~\ref{def_terminal} and $(b)$ \eqref{eq_optimization_problem} is feasible at some time step $k\geq K_h$ then \eqref{eq_optimization_problem} is also feasible at all time steps $\tilde{k} \geq k$.
\end{lemma}
\begin{proof}
	\eqref{eq_optimization_problem} is feasible at all time steps $\tilde{k} \geq k$ if \eqref{eq_optimization_problem} being feasible at time step $k\geq K_h$ implies that, for
	\begin{align}
		u_{k+1+i} = \begin{cases}
			\tilde{u}_{k+1+i-D_{k+1+i}}^{(D_{k+1+i})}& i < D_{k+1+i}\\
			\hat{u}_{k+1}^{(i)}& D_{k+1+i} \leq i \leq N-1\\
			\kappa_{k+1}^{(i)}& i\geq N,
		\end{cases}\label{eq_uki+1}
	\end{align}
	some $(\hat{u}_{k+1}^{(i)})_{\underline{d}\leq i\leq N-1}$ can be computed from $\mathcal{I}_{k+1}$ such that
	\begin{subequations}\label{conditions_feasibility}
		\begin{align}
			&M_{u}u_{k+1+i} \leq n_u,\
			M_{x}x_{k+1+i} \leq n_x \ \ \ \forall i\in [0,N-1]\\
			&x_{k+1+N} \in \mathcal{X}_N(\mathcal{I}_{k+1}).
		\end{align}
	\end{subequations}
	%
	%
	%
	In order to show that such a sequence  $(\hat{u}_{k+1}^{(i)})_{\underline{d}\leq i\leq N-1}$ exists, it is sufficient to show that one such sequence is given by 
	\begin{align}
		\hat{u}_{k+1}^{(i)} = \hat{u}_{k}^{(i+1)}\quad \forall\,  i \leq N-2,\qquad
		\hat{u}_{k+1}^{(N-1)} = \kappa_{k+1}^{(N-1)}.
		\label{eq_u_hat_k+1+i}
	\end{align}
	The resulting sequence $(\hat{u}_{k+1}^{(i)})_{\underline{d}\leq i\leq N-1}$ can be computed from $\mathcal{I}_{k+1}$ since $\hat{u}_{k}^{(i)}$ is computed from $\mathcal{I}_k$ and since $\kappa_{k+1}^{(N-1)} = -L \mathbb{E}\left\{\left.x_{k+N}\right|\mathcal{I}_{k+1}\right\}$ as shown in Appendix~\ref{Appendix_eq_kappa_k+1_N-1}. Therefore, it only remains to be shown that the sequence of actuating variables obtained by inserting \eqref{eq_u_hat_k+1+i} in \eqref{eq_uki+1} satisfies \eqref{conditions_feasibility}.
	
	\subsubsection{Rewritten actuating variables} The sequence of actuating variables resulting from
	\eqref{eq_u_hat_k+1+i} can be written as
	\begin{align}
		u_{k+i} &= \begin{cases}
			\tilde{u}_{k+i-D_{k+i}}^{(D_{k+i})}& i < D_{k+i}\\
			\hat{u}_{k}^{(i)}& D_{k+i} \leq i \leq N-1\\
			\kappa_{k+1}^{(i-1)}& i \geq N
		\end{cases} \label{uki_k+1_rewritten}
	\end{align}
	by $(i)$ inserting \eqref{eq_u_hat_k+1+i} in \eqref{eq_uki+1}, $(ii)$ inserting that $\tilde{u}_{k+1+i-D_{k+1+i}}^{(D_{k+1+i})} = \hat{u}_{k}^{(i+1)}$ for $i = D_{k+1+i}-1$ due to \eqref{eq_uk_tilde} followed by $(iii)$ substituting $i$ for $i+1$.
	

	\subsubsection{Sufficient condition for feasibility}\label{sec_condition_feas}
	Due to \eqref{uki_k+1_rewritten}, the sequence of actuating variables $(u_{k+i})_{i<N}$ is the same as in \eqref{eq_optimization_problem} at time step $k$. Therefore, applying \eqref{uki_k+1_rewritten} also satisfies \eqref{constr_stochastic} and \eqref{terminal_stochastic} which implies that \eqref{eq_7} holds 
	if $u_{k+i}$ can be written in the form \eqref{eq_form_uki_feas} for all ${i \in [N, \hat{N}- 1]}$. Due to \eqref{eq_7.2}, $M_{u}u_{k+\hat{N}} \leq n_u$ and $M_N x_{k+1+\hat{N}} \leq n_N$ also hold if $u_{k+\hat{N}} = -Lx_{k+\hat{N}}$. Summing up,
	\begin{subequations}\label{eq_8}
		\begin{align}
			&M_{u}u_{k+i} \leq n_u, \ M_{x}x_{k+i} \leq n_x \quad \forall i\in [1,\hat{N}]\\
			&M_N x_{k+1+\hat{N}} \leq n_N
		\end{align}
	\end{subequations}
	hold if $u_{k+i}$ can be written in the form \eqref{eq_form_uki_feas} for all ${i \in [N, \hat{N}- 1]}$ and 
	if $u_{k+\hat{N}} = -Lx_{k+\hat{N}}$.
	
	According to Definition~\ref{def_terminal}, \eqref{conditions_feasibility} is satisfied if and only if \eqref{eq_8} holds if $u_{k+i} = \kappa_{k+1}^{(i-1)}$ can be written in the form
	\begin{align}\label{form_kappa_k+1}
		& \kappa_{k+1}^{(i-1)} = \sum\nolimits_{\vec{d}_{k+1} \in \mathcal{D}_{k+1}} p_{k+1,i-1}(\vec{d}_{k+1}) \bar{\kappa}_{k+1}^{(i-1)}(\vec{d}_{k+1})
	\end{align}
	for all $i \in [N+1, \hat{N}]$ where the scalar factors $p_{k+1,i-1}(\vec{d}_{k+1})$ satisfy conditions analogous to \eqref{eq_9}.
	
	Therefore, a sufficient condition for feasibility at time step $k+1$ is feasibility at time step $k$ and that
	\begin{enumerate}
		\item for all $\vec{d}_{k+1} \in \mathcal{D}_{k+1}$ and all $i \in [N, \hat{N}-1]$, there exists a $\vec{d}_k \in \mathcal{D}_k$ such that $\bar{\kappa}_{k+1}^{(i-1)}(\vec{d}_{k+1}) = \bar{\kappa}_{k}^{(i)}(\vec{d}_{k})$
		\item and $\bar{\kappa}_{k+1}^{(\hat{N}-1)}(\vec{d}_{k+1}) = -Lx_{k+\hat{N}}$ for all $\vec{d}_{k+1} \in \mathcal{D}_{k+1}$
	\end{enumerate}
	for $u_{k+i}$ given by \eqref{uki_k+1_rewritten}. This is sufficient since then, $u_{k+i}=\kappa_{k+1}^{(i-1)}$ with $\kappa_{k+1}^{(i-1)}$ of the form \eqref{form_kappa_k+1} can be written in the form \eqref{eq_form_uki_feas} for all ${i \in [N, \hat{N}- 1]}$ and $u_{k+\hat{N}} = \kappa_{k+1}^{(\hat{N}-1)} = -Lx_{k+\hat{N}}$.
	
	
	\subsubsection{Feasibility}
	According to \eqref{kappa_k_bar},
	\begin{align}
		& \bar{\kappa}_{k+1}^{(i-1)}(\vec{d}_{k+1}) = -Lx_{k+i}\, + L \!\! \sum_{j = -\hat{H}_{k+1}^{(i-1)}}^{\overline{d}-1} \!\! A^{i-2-j}B\big(\bar{u}_{k+1}^{(j)}(D_{k+1+j}) -\bar{u}_{k+1}^{(j)}(d_{k+j+1})\big)\nonumber\\
		&= -Lx_{k+i}\, + L \sum\nolimits_{j = -\hat{H}_{k+1}^{(i-1)}+1}^{\overline{d}} A^{i-1-j}B\big(\bar{u}_{k}^{(j)}(D_{k+j}) -\bar{u}_{k}^{(j)}(d_{k+j})\big).\label{eq_kappa_k+1_of_dk+1_vec}
	\end{align}
	since $\bar{u}_{k+1}^{(j-1)}(d_{k+j}) = \bar{u}_{k}^{(j)}(d_{k+j})$ for all $j \leq \overline{d}$ for \eqref{eq_u_hat_k+1+i} due to \eqref{uk_bar}.
	As shown in Appendix~\ref{Appendix_eq_kappa_k+1_of_dk_vec}, this yields
	\begin{align}
		& \bar{\kappa}_{k+1}^{(i-1)}(\vec{d}_{k+1}) = \begin{cases}
			\bar{\kappa}_{k}^{(i)}(\vec{d}_{k})& i \in [N, \hat{N}-1]\\
			-Lx_{k+\hat{N}}& i = \hat{N}
		\end{cases}\label{eq_kappa_k+1_of_dk_vec}
	\end{align}
	with $\vec{d}_{k} \in \mathcal{D}_k$ for any $\vec{d}_{k+1} \in \mathcal{D}_{k+1}$ so the condition for feasibility from Section~\ref{sec_condition_feas} is satisfied. 
\end{proof}

\subsubsection{Admissible Initial States}\label{sec_admissible_x0_stoch}
\begin{definition}[Admissible Initial State]\label{def_admissibe_x0}
	An initial state $x_0$ is considered admissible if it is guaranteed that $M_xA^{k}x_0\leq n_x$ for all $k \in [0, K_h-1]$ and that \eqref{eq_optimization_problem} is feasible at time step $K_h$.
\end{definition}
\begin{lemma}[Satisfied Constraints]\label{le_satisfied_constr}
	If $(a)$ $\kappa_k^{(i)}$ and $\mathcal{X}_N(\mathcal{I}_k)$ are chosen according to Definition~\ref{def_stabilizing}~and~\ref{def_terminal} and $(b)$ $x_0$ is admissible then the constraints \eqref{constraints} are satisfied.
\end{lemma}
\begin{proof}
	For all $k < K_h$, $u_k = 0$ due to \eqref{eq_uk} and \eqref{eq_uk_tilde} so $x_k = A^kx_0$. Therefore, $M_xx_k\leq n_x$ and $M_uu_k\leq n_u$ hold for all $k \in [0,K_h-1]$ if $x_0$ is admissible. The constraints in \eqref{eq_optimization_problem} ensure that $M_{x}x_k \leq n_x$ and $M_{u}u_k \leq n_u$ for all $k\geq K_h$ where \eqref{eq_optimization_problem} is feasible if $x_0$ is admissible due to Lemma~\ref{le_recursive}.
\end{proof}
\begin{lemma}
	Assume that $\kappa_k^{(i)}$ and $\mathcal{X}_N(\mathcal{I}_k)$ are chosen according to Definition~\ref{def_stabilizing}~and~\ref{def_terminal} and let $\tilde{\mathcal{X}}_0(K_h)$ be the set of all $x_0\in\mathbb{R}^n$ for which $(u_{K_h+i})_{\overline{d}\leq i \leq N-1}$ exists such that
	\begin{align}\label{eq_conditions_admissible}
		&M_{u}u_{k} \leq n_u \ \forall k \in [K_h+\overline{d},K_h+N-1],\\
		&M_{x}x_{k} \leq n_x \ \forall k \in [0,K_h+N-1], \quad 
		M_{N}x_{K_h+N} \leq n_N \nonumber
	\end{align}
	if $u_k = 0$ for all $k < K_h + \overline{d}$. Then $(a)$ $x_0$ is admissible if and only if $x_0\in \tilde{\mathcal{X}}_0(K_h)$ with $(b)$ $\tilde{\mathcal{X}}_0(K_h)$ of the form
	\begin{align}
		\tilde{\mathcal{X}}_0(K_h) = \left\{ x_0\in \mathbb{R}^n \big|M_{0,K_h}x_0 \leq n_{0,K_h}\right\}.\label{eq_form_XKh_0}
	\end{align}
\end{lemma}
\begin{proof}
	\eqref{eq_optimization_problem} is guaranteed to be feasible at time step $K_h$ if and only if 
	a sequence $(\hat{u}_{K_h}^{(i)})_{\underline{d}\leq i\leq N-1}$ exists such that
	\begin{subequations}\label{eq_n}
		\begin{align}
			& M_{u}u_{k} \leq n_u, \, M_{x}x_{k} \leq n_x\quad \forall k\in [K_h,K_h+N-1]\\
			& x_{K_h+N} \in \mathcal{X}_N(\mathcal{I}_{K_h})\label{eq_m}
		\end{align}
	\end{subequations}
	is satisfied for all possible realizations of $(D_k)_{k\leq K_h+ \overline{d}}$ for
	\begin{align*}
		u_{K_h+i} = \begin{cases}
			\tilde{u}_{K_h+i-D_{K_h+i}}^{(D_{K_h+i})} = 0& i < D_{K_h+i}\\
			\hat{u}_{K_h}^{(i)}& D_{K_h+i} \leq i \leq N-1\\
			\kappa_{K_h}^{(i)}& i\geq N.
		\end{cases}
	\end{align*}
	Since $D_k = \overline{d}\ \forall k\leq K_h+ \overline{d}$ is possible and yields 
	\begin{align}
		u_{K_h+i} &= \begin{cases}
			\tilde{u}_{K_h+i-\overline{d}}^{(\overline{d})} = 0& i < \overline{d}\\
			\hat{u}_{K_h}^{(i)}& \overline{d} \leq i \leq N-1\\
			\kappa_{K_h}^{(i)}& i\geq N,
		\end{cases}\label{eq_uKh}
	\end{align}
	it is necessary that \eqref{eq_n} holds for \eqref{eq_uKh}. Since choosing $\hat{u}_{K_h}^{(i)} = 0\, \forall i < \overline{d}$ yields \eqref{eq_uKh} for any possible realization of $(D_k)_{k\leq K_h+ \overline{d}}$, this is also sufficient. 
	
	Since \eqref{eq_uKh} yields $\mathbb{E}\left\{\left.u_{K_h+j}\right|\mathcal{I}_{K_h}\right\} = u_{K_h+j} = 0$ for all $j < \overline{d}$, $\kappa_{K_h}^{(i)} = -Lx_{k+i}$ for all $i \geq N$ due to \eqref{eq_kappa_ki_proposed} so the terminal constraint \eqref{eq_m} can be written as $M_Nx_{K_h+N} \leq n_N$.
	
	Summing up, \eqref{eq_optimization_problem} is guaranteed to be feasible at time step $K_h$ if and only if a sequence $(\hat{u}_{K_h}^{(i)})_{\overline{d}\leq i\leq N-1}$ exists such that
	\begin{align*}
		&M_{u}u_{k} \leq n_u \ \forall k \in [K_h+\overline{d},K_h+N-1],\\
		&M_{x}x_{k} \leq n_x \ \forall k \in [K_h,K_h+N-1], \quad 
		M_{N}x_{K_h+N} \leq n_N \nonumber
	\end{align*}
	if $u_k = 0\ \forall\, k < K_h+\overline{d}$. Including the requirement that $M_xA^{k}x_0\leq n_x$ for all $k \in [0, K_h-1]$ where $A^{k}x_0 = x_k$ yields that $x_0$ is admissible if and only if \eqref{eq_conditions_admissible} holds which completes the proof for $(a)$.
	
	If $u_k = 0$ for all $k < K_h + \overline{d}$, $x_k$ with $0 \leq k \leq K_h + N$ can be computed via $x_{k} 
	= A^k x_0 + \underline{B}(K_h,k)\underline{u}$ with
	$\underline{u} = \big[\begin{matrix}
		u_{K_h+N-1}^T& \hdots& u_{K_h+\overline{d}}^T
	\end{matrix}\big]^T$ and $\underline{B}(K_h,k) = \big[\begin{matrix}
		0_{n\times (N+K_h-k)m}\!\!& \! A^0B& \!\!\!\hdots\!\!\!& A^{k-1-K_h-\overline{d}}B
	\end{matrix}\big]$ if $k > K_h\!+\!\overline{d}$ and $\underline{B}(K_h,k) = 0_{n\times (N-\overline{d})m}$ otherwise,
	the set of all $x_0$ for which $(u_{K_h+i})_{\overline{d}\leq i \leq N-1}$ exists such that \eqref{eq_conditions_admissible} holds for $u_k = 0\,\forall k < K_h+\overline{d}$ is the set of all $x_0$ for which $\underline{u}$ exists such that $\mathcal{A}_{K_h}x_0 + \mathcal{B}_{K_h}\underline{u} \leq \mathcal{C}_{K_h}$ with
	\begin{align*}
		\mathcal{A}_{K_h} &= \begin{bmatrix}
			0_{(N-\overline{d})a_u \times n}\\
			M_x A^{0}\\[-1mm]
			\vdots\\
			M_x A^{K_h+N-1}\\
			M_{N} A^{K_h+N}
		\end{bmatrix}, \
		\mathcal{B}_{K_h} = \begin{bmatrix}
			I_{N-\overline{d}} \otimes M_u\\
			M_x\underline{B}(K_h,0)\\[-1mm]
			\vdots\\
			M_x\underline{B}(K_h,K_h\!+\!N\!-\!1)\\
			M_{N}\underline{B}(K_h,N\!+\!K_h)
		\end{bmatrix}\\
		\mathcal{C}_{K_h} &= \begin{bmatrix}
			(1_{N-\overline{d}\times 1} \otimes n_u)^T&
			(1_{K_h+N\times 1} \otimes n_x)^T&
			n_{N}^T
		\end{bmatrix}^T.
	\end{align*}
	This set is can be written in the form \eqref{eq_form_XKh_0} where $M_{0,K_h}$ and $n_{0,K_h}$ can be determined via Algorithm 3.2 in \cite{Keerthi87}. 
\end{proof}

Since $K_h$ is not known in advance, $x_0$ can only be considered to be admissible if $x_0\in \tilde{\mathcal{X}}_0(K_h)$ for all $K_h\in[\underline{h},\overline{h}]$.
\begin{definition}[Set of Admissible Initial States]
	The set admissible initial states $\mathcal{X}_0$ is chosen as
	\begin{align}
		\mathcal{X}_0&= \left\{ x_0\in \mathbb{R}^n \big|M_{0}x_0\leq n_{0}\right\}\label{eq_admissible_x0}
	\end{align}
	with
	$M_{0}^T = \big[\begin{matrix}
		M_{0,\underline{h}}^T&\! \hdots\!& M_{0,\overline{h}}^T
	\end{matrix}\big]$ and $n_{0}^T = \big[\begin{matrix}
		n_{0,\underline{h}}^T&\! \hdots\!& n_{0,\overline{h}}^T
	\end{matrix}\big]$ from \eqref{eq_form_XKh_0} which is the set of all admissible initial states.
\end{definition}

Since $\mathcal{X}_0$ is of the form \eqref{eq_admissible_x0} and since $x_0$ is admissible if $x_0\in\mathcal{X}_0$, property $a)$ in Theorem~\ref{th_1} holds due to Definition~\ref{def_admissibe_x0} and Lemma~\ref{le_recursive}; property $b)$ holds due to Lemma~\ref{le_satisfied_constr}.

\subsection{Convergence}
\begin{lemma}[Convergence with Probability One]\label{le_convergence}
	If $\kappa_k^{(i)}$ and $\mathcal{X}_N(\mathcal{I}_k)$ are chosen according to Definition~\ref{def_stabilizing}~and~\ref{def_terminal} then 
	$\mathbb{P}\big[\lim_{k\to\infty}x_k = 0\big] = 1$ if $x_0$ is admissible.
\end{lemma}
\begin{proof}
	The cost function in \eqref{eq_optimization_problem} can be replaced with 
	$J_k = \mathbb{E}\left\{\left. \bar{J} \right|\mathcal{I}_k\right\}$ where $\bar{J} = \sum_{\tilde{k}=0}^{\infty} \big(x_{\tilde{k}}^TQx_{\tilde{k}} + u_{\tilde{k}}^TRu_{\tilde{k}}\big)$ since $J_k$ is equal to the cost function in \eqref{eq_optimization_problem} except for terms constant with respect to the optimization variables. 
	
	Since $\bar{J}< \infty $ implies that $\lim_{\tilde{k}\to \infty}x_{\tilde{k}} = 0$ and since $\bar{J} \geq 0$, $\mathbb{E}\{ \bar{J}\} < \infty$ implies that $\lim_{\tilde{k}\to \infty}x_{\tilde{k}} = 0$ with probability one. Due to \eqref{kappa_ki_i_geq_hat_N}, $J_k < \infty$ when applying \eqref{eq_uki_stochastic}. Since this holds for any possible $\mathcal{I}_k$, $\mathbb{E}\{ \bar{J}\} = \mathbb{E}\{ J_k\}$ is finite as well so $\lim_{\tilde{k}\to \infty}x_{\tilde{k}} = 0$ with probability one when applying \eqref{eq_uki_stochastic}.
	
	In order to show that $\mathbb{E}\{ \bar{J}\}$ is also finite when solving \eqref{eq_optimization_problem} at each time step $k\geq K_h$, it is sufficient to show that solving \eqref{eq_optimization_problem} at time step $k+1$ and applying 
	\begin{align}
		u_{k+1+i} = \begin{cases}
			\tilde{u}_{k+1+i-D_{k+1+i}}^{(D_{k+1+i})}& i < D_{k+1+i}\\
			\hat{u}_{k+1}^{(i)}& D_{k+1+i} \leq i \leq N-1\\
			\kappa_{k+1}^{(i)}& i\geq N
		\end{cases}\label{uki_k+1}
	\end{align}
	cannot result in a larger value of $J_k$ than applying \eqref{eq_uki_stochastic}. 
	
	Since $J_k = \mathbb{E}\left\{\left. \mathbb{E}\left\{\left. \bar{J} \right|\mathcal{I}_{k+1}\right\} \right|\mathcal{I}_k\right\} = \mathbb{E}\left\{\left. J_{k+1} \right|\mathcal{I}_k\right\}$ due to $\mathcal{I}_k \subseteq \mathcal{I}_{k+1}$, it is sufficient to show that there always exists a feasible sequence $(\hat{u}_{k+1}^{(i)})_{\underline{d}\leq i\leq N-1}$ such that applying \eqref{uki_k+1} does not result in a larger value of $J_{k+1}$ than applying \eqref{eq_uki_stochastic}. It is shown below that one such sequence is given by \eqref{uki_k+1_rewritten}. 
	
	As shown in Appendix~\ref{Appendix_eq_kappa_k+1_of_dk_vec}, \eqref{eq_kappa_k+1_of_dk+1_vec} can be written as \eqref{eq_2} for \eqref{uki_k+1_rewritten}. This can be used to write $u_{k+i} = \kappa_{k+1}^{(i-1)}$ as
	\begin{align}
		u_{k+i} = &-Lx_{k+i} 
		+ L\!\! \sum_{j = -\hat{H}_k^{(i)}}^{\overline{d}-1}\!\! A^{i-1-j}B\left(u_{k+j} - v_{k+j}\right)\label{kappa_stability}
	\end{align}
	with $v_{k+j} = \mathbb{E}\left\{\left.u_{k+j}\right|\mathcal{I}_{k+1}\right\}$ for \eqref{uki_k+1_rewritten}. Due to \eqref{eq_kappa_ki_proposed}, ${u_{k+i} = \kappa_{k}^{(i)}}$ can be written as \eqref{kappa_stability} with ${v_{k+j} = \mathbb{E}\left\{\left.u_{k+j}\right|\mathcal{I}_{k}\right\}}$. 
	
	Since \eqref{uki_k+1_rewritten} and \eqref{eq_uki_stochastic} result in the same $(u_{k+i})_{i < N}$ and $(x_{k+i})_{i \leq N}$, it remains to be shown that applying \eqref{kappa_stability} with ${v_{k+j} = \mathbb{E}\left\{\left.u_{k+j}\right|\mathcal{I}_{k+1}\right\}}$ cannot result in a larger value of 
	\begin{align}
		\mathbb{E}\left\{\left. \sum\nolimits_{\tilde{k}=k+N}^{\infty} \left(x_{\tilde{k}}^TQx_{\tilde{k}} + u_{\tilde{k}}^TRu_{\tilde{k}}\right) \right|\mathcal{I}_{k+1}\right\} \label{eq_cost_stability}
	\end{align}
	than applying \eqref{kappa_stability} with $v_{k+j} = \mathbb{E}\left\{\left.u_{k+j}\right|\mathcal{I}_{k}\right\}$. This is shown in Appendix~\ref{Appendix_opt_pr_stability} by showing that $v_{k+j} = \mathbb{E}\left\{\left.u_{k+j}\right|\mathcal{I}_{k+1}\right\}$ minimizes \eqref{eq_cost_stability} subject to \eqref{kappa_stability}.
\end{proof}
Due to Lemma~\ref{le_convergence}, property $c)$ in Theorem~\ref{th_1} also holds.

\subsection{Rewritten Control Law}
In order to complete the proof for Theorem~\ref{th_1}, it remains to be shown that property $(d)$ holds for the optimization problem \eqref{eq_optimization_problem} with $\kappa_k^{(i)}$ and $\mathcal{X}_N(\mathcal{I}_k)$ from Definition~\ref{def_stabilizing}~and~\ref{def_terminal}. This optimization problem is rewritten by $(i)$ inserting the equality constraint \eqref{eq_uki_stochastic} in the cost function and in the remaining constraints, by $(ii)$ rewriting the terminal constraint \eqref{terminal_stochastic} as inequality constraints and by $(iii)$ evaluating the expected value in \eqref{eq_cost_stochastic} which involves computing the probability distribution of $\vec{D}_k$ from \eqref{eq_Dk_vec}. 

\subsubsection{Rewritten Probabilities}\label{sec_implement_probabilities}
\begin{lemma}[Rewritten Probability Distribution of $\vec{D}_k$]
	The probability distribution of $\vec{D}_k$ from \eqref{eq_Dk_vec} given $\mathcal{I}_k$ can be written as $\mathbb{P}\big[\vec{D}_k=\vec{d}_k \big| \mathcal{I}_k\big]
	= P_k(d_{k-\tilde{H}_k})\vec{P}(\tilde{H}_k, \vec{d}_k)$  where
	\begin{subequations}
		\begin{align}P_k(d_{k-\tilde{H}_k}) &= \mathbb{P}\left[\left.D_{k-\tilde{H}_k} = d_{k-\tilde{H}_k} \right| \mathcal{I}_k\right]\label{eq_Pk}\\
			\vec{P}(\tilde{H}_k, \vec{d}_k)&= \prod\nolimits_{i=-\tilde{H}_k}^{\overline{d}-2}\Phi(d_{k+i},d_{k+i+1}).\label{eq_Pk_vec_rew}
		\end{align}
	\end{subequations}
\end{lemma}
\begin{proof}
	Due to \eqref{eq_hk_tilde}, the controller has not received any $x_{k+i}$ with $i > -\tilde{H}_k$ until time step $k$. Since $\tilde{H}_k < S_k$ for $k \geq K_s$, the controller also has not received any $D_{k+i}$ with $i \geq -\tilde{H}_k$. Since $x_{k+i}$ only depends on $(D_{k+j})_{j < i}$, this implies that $\mathcal{I}_k$ does not contain any data that depends on $(D_{k+i})_{i \geq -\tilde{H}_k}$. Due to \eqref{eq_Markov}, this implies that $\mathbb{P}\big[\vec{D}_k=\vec{d}_k \big| \mathcal{I}_k\big]
	= P_k(d_{k-\tilde{H}_k})\vec{P}(\tilde{H}_k, \vec{d}_k)$ with $P_k(d_{k-\tilde{H}_k})$ from \eqref{eq_Pk} and $\vec{P}(\tilde{H}_k, \vec{d}_k)= \prod\nolimits_{i=-\tilde{H}_k}^{\overline{d}-2}\mathbb{P}\left[\left.D_{k+i+1} = d_{k+i+1} \right| D_{k+i} = d_{k+i}\right]$ which is equal to \eqref{eq_Pk_vec_rew} due to \eqref{eq_Phi}.
\end{proof}

\begin{lemma}[Possible Sequences]
	Let $\tilde{\mathcal{D}}_k \subseteq [\underline{d},\overline{d}]$ denote the set of all possible values of $D_{k-\tilde{H}_k}$ given $\mathcal{I}_k$, i.e. $\tilde{\mathcal{D}}_k = \left\{d\in\mathbb{N}_0 \big|P_k(d) > 0\right\}$. Then, the set $\mathcal{D}_k$ of possible sequences from \eqref{eq_Dk_vec_set} is the set of all $(d_{k+i})_{-\tilde{H}_k \leq i \leq \overline{d}-1}$ with ${d_{k-\tilde{H}_k} \in \tilde{\mathcal{D}}_k}$ for which \eqref{eq_Phi} yields a non-zero probability.
\end{lemma}
\begin{proof}
	This follows directly from the fact that $\mathcal{I}_k$ does not contain data that depends on $(D_{k+i})_{i \geq -\tilde{H}_k}$.
\end{proof}

\begin{lemma}[Expected Values]
	Let $f(\vec{D}_k)$ be some function of $\vec{D}_k$ and let $\hat{\mathcal{D}}_k(\tilde{\mathcal{D}}_k, \delta) \subseteq \mathcal{D}_k(\tilde{\mathcal{D}}_k)$ be the set of all $\vec{d}_k \in \mathcal{D}_k(\tilde{\mathcal{D}}_k)$ with $d_{k-\tilde{H}_k} = \delta$. Then,
	\begin{align}
		&\mathbb{E}\big\{ f(\vec{D}_k) \big|\mathcal{I}_{k}\big\} =
		\!\!\!\sum\limits_{\delta \in \tilde{\mathcal{D}}_k}   P_k(\delta)\!\!\!
		\sum\limits_{\vec{d}_k \in \hat{\mathcal{D}}_k(\tilde{\mathcal{D}}_k, \delta)}\!\!\!\!\!\!\! \vec{P}(\tilde{h}_k, \vec{d}_k) f(\vec{d}_k).\label{eq_expected_values}
	\end{align}
\end{lemma}
\begin{proof}The expected value of $f(\vec{D}_k)$ given $\mathcal{I}_k$ is given by $\mathbb{E}\big\{ f(\vec{D}_k) \big|\mathcal{I}_{k}\big\}
	= \sum_{\vec{d}_k \in \mathcal{D}_k(\tilde{\mathcal{D}}_k)} \mathbb{P}\big[\vec{D}_k=\vec{d}_k \big| \mathcal{I}_k\big] f(\vec{d}_k) = \sum_{\vec{d}_k \in \mathcal{D}_k(\tilde{\mathcal{D}}_k)}P_k(d_{k-\tilde{H}_k})\vec{P}(\tilde{H}_k, \vec{d}_k) f(\vec{d}_k)$ which can be written as \eqref{eq_expected_values}.
\end{proof}

\begin{lemma}[Relevant Information]
	Consider the finite set
	\begin{align}
		& \tilde{\mathcal{I}}_k = \begin{cases}
			\begin{matrix}
				\hspace{-2.45cm}\big\{S_k = s_k, \tilde{H}_k = \tilde{h}_k,\, D_{k-s_k} = d_{k-s_k},\\
				\hspace{2cm} D_{k+i} \in \Theta_{k,i}\, \forall\, i \in [-s_k+1, -\tilde{h}_k-1]\big\}
			\end{matrix}& k \geq K_{sd} 
			\\
			\big\{\tilde{H}_k = \tilde{h}_k,\,  D_{k+i} \in \Theta_{k,i}\, \forall\, i \in [\underline{d}-k, -\tilde{h}_k-1]\big\}
			& k < K_s
			\\
			\big\{S_k = s_k, \tilde{H}_k = \tilde{h}_k,\,  D_{k+i} \in \hat{\Theta}_{k,i}\, \forall\, i \in [\underline{d}-k, -\tilde{h}_k-1]\big\}
			& \text{otherwise}
		\end{cases}\nonumber
	\end{align}
	where $s_k$ and $\tilde{h}_k$ denote known realizations of $S_k$ and $\tilde{H}_k$  and
	\begin{align*}
		\Theta_{k,i} &= \left\{ d_{k+i} \in [\underline{d}, \overline{d}]  \,\big|\, x_{k+i+1} = Ax_{k+i} + B\tilde{u}_{k+i-d_{k+i}}^{(d_{k+i})} \right\}\\
		\hat{\Theta}_{k,i} &= \begin{cases}
			\Theta_{k,i} & -S_k < i \leq -H_k-1\\
			\left\{d\in [\underline{d},\overline{d}] \,\big|\, d \geq k+i
			\right\} & -H_k-1 < i \leq -S_k
			\\
			\left\{d\in\Theta_{k,i} \,\big|\, d \geq k+i
			\right\}& \text{otherwise.}
		\end{cases}
	\end{align*}
	This set is contained in $\mathcal{I}_k$ and contains all information that is relevant for evaluating \eqref{eq_Pk} in the sense that
	\begin{align}
		&P_k(\delta) = 
		\mathbb{P}\left[\left.D_{k-\tilde{H}_k} = \delta \right| \tilde{\mathcal{I}}_k\right] \quad \forall\, k\geq K_h. \label{eq_relevant_information}
	\end{align}
\end{lemma}
\begin{proof}
	\eqref{eq_relevant_information} is shown by determining all information about the realization of $\mathbf{D}$ contained in $\mathcal{I}_k$ that is relevant for computing the probability distribution of $D_{k-\tilde{H}_k}$.
	
	For any $k \geq K_h$, $H_k$ and $(x_{k+i})_{-k \leq i \leq -H_k}$ are known. While $u_{k+i} = 0$ for $k+i < \underline{d}$ for any $D_{k+i}$, $u_{k+i} = \tilde{u}_{k+i-d_{k+i}}^{(d_{k+i})}$ 
	might not satisfy $x_{k+i+1} = Ax_{k+i} + Bu_{k+i}$ for all $d_{k+i} \in [\underline{d}, \overline{d}]$ for larger values of $i$. This yields
	\begin{align}
		&D_{k+i} \in \Theta_{k,i} \quad \forall\, i \in [\underline{d}-k, -H_k-1]\label{information_sensors}.
	\end{align}
	
	If $k < K_s$, all relevant information is given by $H_k$ and \eqref{information_sensors} where $\tilde{H}_k = H_k$ due to \eqref{eq_hk_tilde}. If $K_{sd} > k \geq K_{s}$, $S_k$ is known in addition to \eqref{information_sensors} and it is known that $D_{\tilde{k}} > \tilde{k}\ \forall\, \tilde{k} \leq k-S_k$. Since $D_{\tilde{k}} > \tilde{k}$ always holds for $\tilde{k}<\underline{d}$, all relevant information is given by $H_k$, $S_k$ and $D_{k+i} \in \hat{\Theta}_{k,i} \ \forall\, i \in [\underline{d}-k, -\tilde{H}_k-1]$.
	
	If $k \geq K_{sd}$, $D_{k-S_k}$ is known and $-S_k < -\tilde{H}_k$ due to \eqref{eq_hk_tilde}. Due to \eqref{eq_Markov}, all other information about $D_{k+i}$ with $i \leq -S_k$ is not relevant; all relevant information is given by $S_k$, $H_k$, $D_{k-S_k}$ and $D_{k+i} \in \Theta_{k,i}$ for all $i \in [-S_k+1, -H_k-1]$ where $[-S_k\! +\! 1, -H_k\! -\! 1] = [-S_k\! +\! 1, -\tilde{H}_k\! -\! 1]$ due to \eqref{eq_hk_tilde}.
\end{proof}

As shown in Appendix~\ref{Appendix_eq_Pk}, \eqref{eq_relevant_information} with $k \geq K_h$ can be computed via
\begin{subequations}
	\begin{align}
		&P_k(\delta) = \begin{cases}
			\Phi(d_{k-s_k}, \delta)
			& 
			k \geq K_{sd}\, \land\,
			\tilde{h}_k = s_k-1
			\\
			\begin{matrix}
				\!\!\!\!\!\!\!\!\!\!\!\!\!\!\!\!\!\! \tilde{P}_2(k-s_k,k-\tilde{h}_k,d_{k-s_k},\delta,\\
				\qquad\qquad \quad  (\Theta_{k,i})_{-s_k+1 \leq i \leq -\tilde{h}_k - 1})
			\end{matrix}
			& 
			k \geq K_{sd}\, \land\,
			\tilde{h}_k < s_k-1
			\\
			\sum\nolimits_{d_0 = \underline{d}}^{\overline{d}}
			\mu(d_0)
			\Phi_{k-\tilde{h}_k}(d_0,\delta)
			&
			K_{sd}> k \, \land\,
			k-\tilde{h}_k \leq \underline{d}
			\\
			\tilde{P}_4(\underline{d},k-\tilde{h}_k, \delta, (\hat{\Theta}_{k,i})_{\underline{d}-k \leq i \leq -\tilde{h}_k - 1})
			&
			K_{sd}> k \geq K_s\,
			\land \, k-\tilde{h}_k > \underline{d}
			\\
			\tilde{P}_4(\underline{d},k-\tilde{h}_k, \delta, (\Theta_{k,i})_{\underline{d}-k \leq i \leq -\tilde{h}_k - 1})
			&
			k < K_s\,
			\land \, k-\tilde{h}_k > \underline{d}
		\end{cases}\nonumber\\
		&\tilde{P}_1(\underline{i},\overline{i},\delta_{\underline{i}}, (\Theta_i)_{\underline{i}+1 \leq i \leq \overline{i}})
		=\sum\nolimits_{\delta_{\underline{i}+1} \in \Theta_{\underline{i}+1}}
		\!\! \hdots \sum\nolimits_{\delta_{\overline{i}}\, \in \Theta_{\overline{i}}} \prod\nolimits_{i = \underline{i}}^{\overline{i}-1} \Phi(\delta_{i}, \delta_{i+1})
		\label{eq_P1_tilde}\\
		&\tilde{P}_2(\underline{i},\overline{i},\delta_{\underline{i}}, \delta_{\overline{i}}, (\Theta_i)_{\underline{i}+1 \leq i \leq \overline{i} - 1})
		=
		\frac{\tilde{P}_1(\underline{i},\overline{i},\delta_{\underline{i}}, (\Theta_i)_{\underline{i}+1 \leq i \leq \overline{i}-1}\,\!^\frown \{\delta_{\overline{i}}\})}
		{\tilde{P}_1(\underline{i},\overline{i}-1,\delta_{\underline{i}}, (\Theta_i)_{\underline{i}+1 \leq i \leq \overline{i}-1})}\label{eq_P2_tilde}
			\\
		&\tilde{P}_3(\underline{i},\overline{i}, (\Theta_i)_{\underline{i} \leq i \leq \overline{i}})
		=
		\begin{cases}
			\sum\nolimits_{d_0 \in \Theta_{0}}
			\mu(d_0)
			& 0 = \underline{i} = \overline{i}
			\\
			\sum\nolimits_{d_0 \in \Theta_{0}}
			\mu(d_0)
			\tilde{P}_1(0,\overline{i},d_0, (\Theta_i)_{1 \leq i \leq \overline{i}})
			& 0 = \underline{i} < \overline{i}
			\\
			\sum\nolimits_{d_0 = \underline{d}}^{\overline{d}}
			\mu(d_0)
			\sum\nolimits_{d_{\underline{i}} \in \Theta_{\underline{i}}}
			\Phi_{\underline{i}}(d_0,d_{\underline{i}})
			& 0 < \underline{i} = \overline{i}
			\\
			\begin{matrix}
				\sum\nolimits_{d_0 = \underline{d}}^{\overline{d}}
				\mu(d_0)
				\sum\nolimits_{d_{\underline{i}} \in \Theta_{\underline{i}}}
				\Phi_{\underline{i}}(d_0,d_{\underline{i}})
				\\ \quad
				\tilde{P}_1(\underline{i},\overline{i},d_{\underline{i}}, (\Theta_i)_{\underline{i}+1 \leq i \leq \overline{i}})
			\end{matrix}
			& 0 < \underline{i} < \overline{i}
		\end{cases}\label{eq_P3_tilde}\\
		&\tilde{P}_4(\underline{i},\overline{i}, \delta_{\overline{i}}, (\Theta_i)_{\underline{i} \leq i \leq \overline{i} - 1}) = 
		\frac{\tilde{P}_3(\underline{i},\overline{i}, (\Theta_i)_{\underline{i} \leq i \leq \overline{i}-1}\,\!^\frown \{\delta_{\overline{i}}\})}
		{\tilde{P}_3(\underline{i},\overline{i}-1, (\Theta_i)_{\underline{i} \leq i \leq \overline{i}-1})}.
		\label{eq_P4_tilde}
	\end{align}
\end{subequations}

\subsubsection{Extended State Vector}
For $k \geq K_h$, $x_{k-\tilde{H}_k}$ can be computed from $\mathcal{I}_k$ since $x_{k-H_k}$ is known and $(u_{k+i})_{-H_k\leq i \leq -S_k}$ can be determined via \eqref{eq_uk}. Therefore, the extended state vector
$\hat{x}_k = \big[\begin{matrix}
	x_{k-\tilde{H}_k}^T & \tilde{u}_{k-1}^T & \hdots & \tilde{u}_{k-\overline{d}-\overline{h}}^T 
\end{matrix}\big] \in \mathbb{R}^{\hat{n}}$ with $\hat{n} = n + \tilde{m}(\overline{d}+\overline{h})$ is known at time step $k\geq K_h$. Using $\hat{u}_k = \big[\begin{matrix}
	\hat{u}_k^{(N-1)^T} & \hdots & \hat{u}_k^{(\underline{d})^T}
\end{matrix}\big]^T\in \mathbb{R}^{\hat{m}}$
with $\hat{m} = m(N-\underline{d})$, $u_{k+i}$ resulting from \eqref{eq_uki_stochastic} can be written as
\begin{align}
	u_{k+i} &= \begin{cases}
		\tilde{u}_{k+i-D_{k+i}}^{(D_{k+i})}  & -\tilde{H}_k\leq i < D_{k+i} \leq \overline{d}\\
		\hat{u}_k^{(i)}  &  D_{k+i} \leq i \leq N-1
	\end{cases}
	\label{eq_uki_extended_state_1}\\
	&= \begin{cases}
		\bar{T}(i,D_{k+i})\hat{x}_k + \tilde{T}(i,D_{k+i})\hat{u}_k& -\tilde{H}_k \leq i \leq \overline{d}-1\\
		\hat{T}(i)\hat{u}_k & \overline{d} \leq i \leq N-1
	\end{cases}\nonumber
\end{align}
with $\hat{T}(i)= \begin{bmatrix}
	0_{m\times (N-1-i)m}& I_m& 0_{m\times \hat{m} - (N-i)m}
\end{bmatrix}$, $\tilde{T}(i,D_{k+i})=\hat{T}(i)$ if $D_{k+i} \leq i$ and $\tilde{T}(i,D_{k+i})=0_{m\times \hat{m}}$ otherwise and
\begin{align*}\bar{T}(i,D_{k+i})= \begin{cases}
		0_{m\times \hat{n}} & D_{k+i} \leq i\\
		 \begin{bmatrix}
	0_{m\times n + (D_{k+i}-i-1)\tilde{m}}^T\\
	\begin{bmatrix}
		0_{m \times (\overline{d}-D_{k+i})m} & I_m & 0_{m \times (D_{k+i}-\underline{d})m}
	\end{bmatrix}^T\\
	0_{m\times \hat{n} - (n + (D_{k+i}-i)\tilde{m})}^T
\end{bmatrix}^T & \text{otherwise.}
\end{cases}\end{align*}

Using $x_{k+i} = A^{i+\tilde{H}_k} x_{k-\tilde{H}_k} + \sum\nolimits_{j=-\tilde{H}_k}^{i-1}A^{i-1-j}Bu_{k+j}$ for $i \geq -\tilde{H}_k$ and \eqref{eq_uki_extended_state_1}, $x_{k+i}$ with $-\tilde{H}_k \leq i \leq N$ is written as
\begin{align}
	&x_{k+i} = \hat{A}(i,\tilde{H}_k,\vec{D}_k)\hat{x}_k 
	+ \hat{B}(i,\tilde{H}_k,\vec{D}_k)\hat{u}_k\label{eq_xki_extended_state_2}
\end{align}
with $\hat{B}(i,\tilde{H}_k,\vec{D}_k)= \sum\nolimits_{j=\overline{d}}^{i-1}A^{i-1-j}B\hat{T}(j) + \sum\nolimits_{j=-\tilde{H}_k}^{\min(i,\overline{d})-1}A^{i-1-j}B\tilde{T}(j,D_{k+j})$ and $\hat{A}(i,\tilde{H}_k,\vec{D}_k)= \begin{bmatrix}
	A^{i+\tilde{H}_k} & 0_{n\times \hat{n}-n}
\end{bmatrix} + \sum\nolimits_{j=-\tilde{H}_k}^{\min(i,\overline{d})-1}A^{i-1-j}B\bar{T}(j,D_{k+j})$
and $u_{k+i}$ with $-\tilde{H}_k \leq i \leq N-1$ can be written as
\begin{align}
	&u_{k+i} = \hat{A}_u(i,\tilde{H}_k,\vec{D}_k)\hat{x}_k 
	+ \hat{B}_u(i,\tilde{H}_k,\vec{D}_k)\hat{u}_k \label{eq_uki_extended_state_2}
\end{align}
where $\hat{A}_u(i,\tilde{H}_k,\vec{D}_k) = \bar{T}(i,D_{k+i})$ and $\hat{B}_u(i,\tilde{H}_k,\vec{D}_k) = \tilde{T}(i,D_{k+i})$ for $i \leq \overline{d}-1$; for $i \geq \overline{d}$, $\hat{A}_u(i,\tilde{H}_k,\vec{D}_k) = 0_{m\times \hat{n}}$ and $\hat{B}_u(i,\tilde{H}_k,\vec{D}_k) = \hat{T}(i)$.
The stabilizing control law is given by
$\kappa_k^{(i)} = \sum\nolimits_{\vec{d}_k \in \mathcal{D}_k} \mathbb{P}\big[\vec{D}_k=\vec{d}_k \big| \mathcal{I}_k\big] \bar{\kappa}_{k}^{(i)}(\vec{d}_k)\ \forall i \geq N$ where, due to \eqref{kappa_k_bar} and \eqref{eq_uki_extended_state_1},
\begin{align}
	&\bar{\kappa}_{k}^{(i)}(\vec{d}_k) = -Lx_{k+i}
	+ \hat{L}_x(i,\tilde{H}_k,\vec{D}_k,\vec{d}_k)\hat{x}_k + \hat{L}_u(i,\tilde{H}_k,\vec{D}_k,\vec{d}_k)\hat{u}_k \label{eq_kappa_bar_ki_extended_state_1}
\end{align}
with $\hat{L}_x(i,\tilde{H}_k,\vec{D}_k,\vec{d}_k) 
= L \sum\nolimits_{j = -\hat{H}_k^{(i)}}^{\overline{d}-1}\!\! A^{i-1-j}B $ $\left(\bar{T}(j,D_{k+j}) -\bar{T}(j,d_{k+j})\right)$ and
$\hat{L}_u(i,\tilde{H}_k,\vec{D}_k,\vec{d}_k)
=L \sum\nolimits_{j = -\hat{H}_k^{(i)}}^{\overline{d}-1}\!\! A^{i-1-j}B\big(\tilde{T}(j,D_{k+j}) -\tilde{T}(j,d_{k+j})\big)$.

\subsubsection{Rewritten Constraints}\label{sec_implement_constr}
\begin{lemma}\label{le_iteration_terminal}
	
	Consider $x_{k+i+1} = Ax_{k+i} + Bu_{k+i}$, $\hat{y}_k = \begin{bmatrix}
		\hat{x}_k^T & \hat{u}_k^T
	\end{bmatrix}^T$ and constraints of the form
	\begin{align}
		& M_x x_{k+i} \leq n_x, \quad  M_u u_{k+i} \leq n_u, \label{eq_constr_terminal_step1}\\
		&\bar{M}_x(i \! + \! 1,\tilde{\mathcal{D}}_k)x_{k+i+1} \! + \! \bar{M}_y(i \! + \! 1,\tilde{H}_k,\vec{D}_k,\tilde{\mathcal{D}}_k)\hat{y}_k \!\leq\! \bar{n}(i \! + \! 1,\tilde{\mathcal{D}}_k)\nonumber
	\end{align}
	where $N \leq i < \hat{N}$. These constraints are satisfied for all $u_{k+i}$ of the form \eqref{eq_form_uki_feas} if and only if, for all $\vec{d}_k \in \mathcal{D}_k(\tilde{\mathcal{D}}_k)$,
	\begin{align}
		&\tilde{\bar{M}}_x(i,\tilde{\mathcal{D}}_k)x_{k+i} + \tilde{\bar{M}}_y(i,\tilde{H}_k,\vec{D}_k,\tilde{\mathcal{D}}_k,\vec{d}_k)\hat{y}_k \leq \tilde{\bar{n}}(i,\tilde{\mathcal{D}}_k)\label{eq_constraint_singe_dk_hat}\\
		&\tilde{\bar{n}}(i,\tilde{\mathcal{D}}_k) = \begin{bmatrix}
			\bar{n}^T(i+1,\tilde{\mathcal{D}}_k)&
			n_x^T&
			n_u^T
		\end{bmatrix}^T,\nonumber
		\\
		&\tilde{\bar{M}}_x(i,\tilde{\mathcal{D}}_k) = \begin{bmatrix}
			\bar{M}_x(i+1,\tilde{\mathcal{D}}_k)(A-BL)\\
			M_x\\
			-M_uL
		\end{bmatrix}, \nonumber
	\\
	  &\tilde{\bar{M}}_y(i,\tilde{H}_k,\vec{D}_k,\tilde{\mathcal{D}}_k,\vec{d}_k) =
		\begin{bmatrix}
			\bar{M}_x(i\!+\!1,\tilde{\mathcal{D}}_k)B\hat{L}(i,\tilde{H}_k,\vec{D}_k,\vec{d}_k)
			\!+\! \bar{M}_y(i\!+\!1,\tilde{H}_k,\vec{D}_k,\tilde{\mathcal{D}}_k)\\
			0_{a_x,n}\\
			M_u\hat{L}(i,\tilde{H}_k,\vec{D}_k,\vec{d}_k)
		\end{bmatrix}\!\!.
		\nonumber
	\end{align}
\end{lemma}
\begin{proof}
	Because the constraints \eqref{eq_constr_terminal_step1} are linear in $u_{k+i}$, they are satisfied for all $u_{k+i}$ of the form \eqref{eq_form_uki_feas} if and only if they are satisfied for $u_{k+i} =\bar{\kappa}_{k}^{(i)}(\vec{d}_k)$ for all $\vec{d}_k \in \mathcal{D}_k(\tilde{\mathcal{D}}_k)$. 
	
	Due to \eqref{eq_kappa_bar_ki_extended_state_1}, the following holds for for $u_{k+i} =\bar{\kappa}_{k}^{(i)}(\vec{d}_k)$:
	\begin{subequations}
		\begin{align*}
			u_{k+i} &= -Lx_{k+i}
			+ \hat{L}(i,\tilde{H}_k,\vec{D}_k,\vec{d}_k)\hat{y}_k\\
			x_{k+i+1} &= (A-BL)x_{k+i} + B\hat{L}(i,\tilde{H}_k,\vec{D}_k,\vec{d}_k)\hat{y}_k\\
			\hat{L}(i,\tilde{H}_k,\vec{D}_k,\vec{d}_k) &= \begin{bmatrix}
				\hat{L}_x(i,\tilde{H}_k,\vec{D}_k,\vec{d}_k)
				& 
				\hat{L}_u(i,\tilde{H}_k,\vec{D}_k,\vec{d}_k)
			\end{bmatrix}.
		\end{align*}
	\end{subequations}
	This can be inserted in \eqref{eq_constr_terminal_step1} which yields \eqref{eq_constraint_singe_dk_hat}.
\end{proof}
\begin{lemma}[Terminal Constraints]
	Let $\vec{d}_k^{\, (j)}$ be the $j$th element of the set $\mathcal{D}_k(\tilde{\mathcal{D}}_k)$, let $\overline{j}_D$ be the number of elements in this set and consider the constraints 
	\begin{align}
		\bar{M}_x(N,\tilde{\mathcal{D}}_k)x_{k+N}\! +\! \bar{M}_y(N,\tilde{H}_k,\vec{D}_k,\tilde{\mathcal{D}}_k)\hat{y}_k \!\leq\! \bar{n}(N,\tilde{\mathcal{D}}_k)\label{eq_temp_2309u5}
	\end{align}
	where $\bar{M}_x$, $\bar{M}_y$ and $\bar{n}$ are obtained by $(a)$ starting with 
	\begin{align}\label{eq_terminal_constr_N_hat}
		&\bar{M}_x(\hat{N},\tilde{\mathcal{D}}_k)x_{k+\hat{N}} \! +\! \bar{M}_y(\hat{N},\tilde{H}_k,\vec{D}_k,\tilde{\mathcal{D}}_k)\hat{y}_k \!\leq\! \bar{n}(\hat{N},\tilde{\mathcal{D}}_k)
	\end{align}
	where $\bar{n}(\hat{N},\tilde{\mathcal{D}}_k)= n_N$, $\bar{M}_x(\hat{N},\tilde{\mathcal{D}}_k) = M_{N}$, and $\bar{M}_y(\hat{N},\tilde{H}_k,\vec{D}_k,\tilde{\mathcal{D}}_k) = 0$
	followed by (b) applying 
	\begin{align}
		&\bar{M}_x(i,\tilde{\mathcal{D}}_k)x_{k+i} + \bar{M}_y(i,\tilde{H}_k,\vec{D}_k,\tilde{\mathcal{D}}_k)\hat{y}_k \leq \bar{n}(i,\tilde{\mathcal{D}}_k)\nonumber\\
		&\bar{M}_y(i,\tilde{H}_k,\vec{D}_k,\tilde{\mathcal{D}}_k) = \begin{bmatrix}
			\tilde{\bar{M}}_y(i,\tilde{H}_k,\vec{D}_k,\tilde{\mathcal{D}}_k,\vec{d}^{\, (1)}_k)\\[-1mm]
			\vdots\\
			\tilde{\bar{M}}_y(i,\tilde{H}_k,\vec{D}_k,\tilde{\mathcal{D}}_k,\vec{d}^{\, (\overline{j}_D)}_k)
		\end{bmatrix},
		\label{eq_terminal_iteration}\\
		&\bar{M}_x(i,\tilde{\mathcal{D}}_k) = 1_{\overline{j}_D} \otimes \tilde{\bar{M}}_x
		,\quad \bar{n}(i,\tilde{\mathcal{D}}_k) = 1_{\overline{j}_D} \otimes \tilde{\bar{n}}(i,\tilde{\mathcal{D}}_k)\nonumber
	\end{align}
	until $i=N$. Then, \eqref{terminal_stochastic} with $\mathcal{X}_N(\mathcal{I}_k)$ from Definition~\ref{def_terminal} holds for all $\vec{d}_k \in \mathcal{D}_k$  if and only if the constraints \eqref{eq_temp_2309u5} are satisfied.
\end{lemma}
\begin{proof}
	According to Definition~\ref{def_terminal}, \eqref{eq_7} has to hold where $M_{N}x_{k+\hat{N}} \leq n_N$ can be written as \eqref{eq_terminal_constr_N_hat}. A sufficient and necessary condition for \eqref{eq_7} can be determined using Lemma~\ref{le_iteration_terminal} as follows: $(a)$ start with $i=\hat{N}-1$, $(b)$ write down \eqref{eq_constr_terminal_step1} where $\bar{M}_x(i \! + \! 1,\tilde{\mathcal{D}}_k)$, $\bar{M}_y(i \! + \! 1,\tilde{H}_k,\vec{D}_k,\tilde{\mathcal{D}}_k)$ and $\bar{n}(i \! + \! 1,\tilde{\mathcal{D}}_k)$ are given by \eqref{eq_terminal_constr_N_hat} for $i=\hat{N}-1$, $(c)$ use the fact that \eqref{eq_constraint_singe_dk_hat} holds for all $\vec{d}_k \in \mathcal{D}_k(\tilde{\mathcal{D}}_k)$ if and only if \eqref{eq_terminal_iteration} holds to obtain $\bar{M}_x(i,\tilde{\mathcal{D}}_k)$, $\bar{M}_y(i,\tilde{H}_k,\vec{D}_k,\tilde{\mathcal{D}}_k)$ and $\bar{n}(i,\tilde{\mathcal{D}}_k)$ and $(d)$ decrease $i$ by one and go back to $(b)$ until $i = N$. This results in \eqref{eq_temp_2309u5}.
\end{proof}
Using \eqref{eq_xki_extended_state_2},  \eqref{eq_temp_2309u5} can also be written as
\begin{align}
	\hat{M}_{x,N}(\tilde{h}_k,\vec{D}_k,\tilde{\mathcal{D}}_k)&\hat{x}_k \!+\! \hat{M}_{u,N}(\tilde{h}_k,\vec{D}_k,\tilde{\mathcal{D}}_k)\hat{u}_k \!\leq\! \hat{n}_{N}(\tilde{\mathcal{D}}_k) \label{eq_terminal_constr}
\end{align}
with
$\hat{M}_{x,N}(\tilde{h}_k,\vec{D}_k,\tilde{\mathcal{D}}_k) = \bar{M}_x(N,\tilde{\mathcal{D}}_k)\hat{A}(N,\tilde{h}_k,\vec{D}_k)
+ \bar{M}_y(N,\tilde{h}_k,\vec{D}_k,\tilde{\mathcal{D}}_k)\begin{bmatrix}
	I_{\hat{n}} & 0_{\hat{m} \times\hat{n}^T}
\end{bmatrix}^T$,\\
$\hat{M}_{u,N}(\tilde{h}_k,\vec{D}_k,\tilde{\mathcal{D}}_k) = \bar{M}_x(N,\tilde{\mathcal{D}}_k)\hat{B}(N,\tilde{h}_k,\vec{D}_k)
+ \bar{M}_y(N,\tilde{h}_k,\vec{D}_k,\tilde{\mathcal{D}}_k)\begin{bmatrix}
	0_{\hat{n} \times \hat{m}}^T  &  I_{\hat{m}}
\end{bmatrix}^T$ and $\hat{n}_N(\tilde{\mathcal{D}}_k) = \bar{n}(N,\tilde{\mathcal{D}}_k)$.

\begin{lemma}[Terminal- and State Constraints]
	Consider
	\begin{subequations}
		\begin{align}
			&\tilde{\vec{M}}_{x}(\tilde{H}_k,\tilde{\mathcal{D}}_k)\hat{x}_k + \tilde{\vec{M}}_{u}(\tilde{H}_k,\tilde{\mathcal{D}}_k)\hat{u}_k \leq \tilde{\vec{n}}(\tilde{\mathcal{D}}_k) \label{state_constr_rewritten}
			\\
			&\tilde{\vec{M}}_{x}(\tilde{H}_k,\tilde{\mathcal{D}}_k) = \!\begin{bmatrix}
				\hat{\vec{M}}_{x}(\tilde{H}_k,\vec{D}_k,\vec{d}_k^{\,(1)})\\[-1mm]
				\vdots\\
				\hat{\vec{M}}_{x}(\tilde{H}_k,\vec{D}_k,\vec{d}_k^{\, (\overline{j}_D)})
			\end{bmatrix}\!\! , \  \tilde{\vec{M}}_{u}(\tilde{H}_k,\vec{D}_k,\tilde{\mathcal{D}}_k) = 
			\!\begin{bmatrix}
				\hat{\vec{M}}_{u}(\tilde{H}_k,\vec{D}_k,\vec{d}_k^{\,(1)})\\[-1mm]
				\vdots\\
				\hat{\vec{M}}_{u}(\tilde{H}_k,\vec{D}_k,\vec{d}_k^{\, (\overline{j}_D)})
			\end{bmatrix}
			\nonumber\\
			&  \tilde{\vec{n}}(\tilde{\mathcal{D}}_k) = 1_{\overline{j}_D} \otimes \hat{\vec{n}}(\tilde{\mathcal{D}}_k)\nonumber
			\nonumber
		\end{align}
		where 
		\begin{align} \label{state_constr_temp}
			&\hat{\vec{M}}_{x}(\tilde{H}_k,\vec{D}_k,\tilde{\mathcal{D}}_k) = \begin{bmatrix}
				\hat{M}_{x,N}(\tilde{H}_k,\vec{D}_k,\tilde{\mathcal{D}}_k)\\
				M_x\hat{A}(\underline{d}+1,\tilde{H}_k,\vec{D}_k)\\[-1mm]
				\vdots\\
				M_x\hat{A}(N-1,\tilde{H}_k,\vec{D}_k)
			\end{bmatrix}
			\\
			&\hat{\vec{M}}_{u}(\tilde{H}_k,\vec{D}_k,\tilde{\mathcal{D}}_k) = \begin{bmatrix}
				\hat{M}_{u,N}(\tilde{H}_k,\vec{D}_k,\tilde{\mathcal{D}}_k)\\
				M_x\hat{B}(\underline{d}+1,\tilde{H}_k,\vec{D}_k)\\[-1mm]
				\vdots\\
				M_x\hat{B}(N-1,\tilde{H}_k,\vec{D}_k)
			\end{bmatrix}, \
			\hat{\vec{n}}(\tilde{\mathcal{D}}_k) = \begin{bmatrix}
				\hat{n}_{N}^T(\tilde{\mathcal{D}}_k)\\
				n_x\\
				\vdots\\
				n_x
			\end{bmatrix}.\nonumber
		\end{align}
	\end{subequations}
	The terminal constraints \eqref{terminal_stochastic} with the $\mathcal{X}_N(\mathcal{I}_k)$ from Definition~\ref{def_terminal} and the state constraints $M_{x}x_{k+i} \leq n_x\ \forall\, i\in [0,N-1]$ in \eqref{constr_stochastic} are satisfied for all $\vec{d}_k \in \mathcal{D}_k$ if and only if \eqref{state_constr_rewritten} holds.
\end{lemma}
\begin{proof}
	Since $(u_{k+i})_{i < \underline{d}}$ and $(x_{k+i})_{i \leq \underline{d}}$ do not depend on the optimization variables, only $M_{x}x_{k+i} \leq n_x \, \forall i\in [0,\underline{d}]$ and \eqref{eq_terminal_constr} have to be considered. This can be written as
	\begin{align}\label{eq_all_constr_Dk_vec}
		&\hat{\vec{M}}_{x}(\tilde{H}_k,\vec{D}_k,\tilde{\mathcal{D}}_k)\hat{x}_k + \hat{\vec{M}}_{u}(\tilde{H}_k,\vec{D}_k,\tilde{\mathcal{D}}_k)\hat{u}_k \leq \hat{\vec{n}}(\tilde{\mathcal{D}}_k)
	\end{align}
	with \eqref{state_constr_temp} using \eqref{eq_xki_extended_state_2}. This has to be satisfied for all possible $\vec{D}_k$, i.e. all $\vec{D}_k\in \mathcal{D}_k(\tilde{\mathcal{D}}_k)$, which can be written as \eqref{state_constr_rewritten}.
\end{proof}

\begin{lemma}[All Constraints]\label{le_all_constraints}
	The terminal constraints \eqref{terminal_stochastic} with $\mathcal{X}_N(\mathcal{I}_k)$ from Definition~\ref{def_terminal}, all constraints on the state vector and all constraints on the actuating variables in \eqref{constr_stochastic} that depend on $\hat{u}_k$ can be written as 
	\begin{align}
		&\vec{M}_{x}(\tilde{H}_k,\tilde{\mathcal{D}}_k)\hat{x}_k + \vec{M}_{u}(\tilde{H}_k,\tilde{\mathcal{D}}_k)\hat{u}_k \leq \vec{n}(\tilde{\mathcal{D}}_k)\label{eq_constraints_redundant}
		\\
		&\vec{M}_{x}(\tilde{H}_k,\tilde{\mathcal{D}}_k) = \begin{bmatrix}
			\tilde{\vec{M}}_{x}(\tilde{H}_k,\tilde{\mathcal{D}}_k)\\
			0_{(N-\underline{d})a_u \times \hat{n}}
		\end{bmatrix}
		, \
		\vec{M}_{u}(\tilde{H}_k,\tilde{\mathcal{D}}_k) = \begin{bmatrix}
			\tilde{\vec{M}}_{u}(\tilde{H}_k,\tilde{\mathcal{D}}_k)\\
			\big(I_{N-\underline{d}} \otimes M_u\big)
		\end{bmatrix},
		\nonumber
	\end{align}
	\begin{align*}
		&\vec{n}(\tilde{\mathcal{D}}_k) = \big[\begin{matrix}
			\tilde{\vec{n}}(\tilde{\mathcal{D}}_k)^T&
			\big(1_{N-\underline{d}} \otimes n_u\big)^T
		\end{matrix}\big]^T.\nonumber
	\end{align*}
\end{lemma}
\begin{proof}
	All actuating variable constraints in \eqref{eq_optimization_problem} that depend on $\hat{u}_k$ are given by $M_u\hat{u}_k^{(i)} \leq n_u\, \forall i\in [\underline{d},N-1]$. These constraints and \eqref{state_constr_rewritten} can be written as \eqref{eq_constraints_redundant}.
\end{proof}

\subsubsection{ Rewritten Cost Function}
The cost function in \eqref{eq_optimization_problem} can be replaced with
$\hat{J}_k = \mathbb{E}\Big\{\sum\nolimits_{i=\underline{d}+1}^{\infty}x_{k+i}^TQx_{k+i} + \sum\nolimits_{i=\underline{d}}^{\infty}u_{k+i}^TRu_{k+i} \Big|\mathcal{I}_k\Big\}$
since $(u_{k+i})_{i < \underline{d}}$ and $(x_{k+i})_{i \leq \underline{d}}$ do not depend on $\hat{u}_k$. Since $u_{k+i} = \kappa_{k}^{(i)} = -Lx_{k+i}$ for all $i\geq \hat{N}$, this can be written as
\begin{align}
	&\hat{J}_k = \mathbb{E}\big\{ \hat{J}^{1}_k + \hat{J}^{2}_k + \hat{J}^{3}_k + \hat{J}^{4}_k \big|\mathcal{I}_k\big\} \label{eq_Jk_hat}
\end{align}
with $\hat{J}^{1}_k = \sum\nolimits_{i=\underline{d}}^{N-1}u_{k+i}^TRu_{k+i}$, $\hat{J}^{2}_k = \sum\nolimits_{i=\underline{d}+1}^{N}x_{k+i}^TQx_{k+i}$, $\hat{J}^{3}_k = \sum\nolimits_{i=N}^{\hat{N}-1}u_{k+i}^TRu_{k+i}$ and\\ $\hat{J}^{4}_k = \sum\nolimits_{i=N+1}^{\hat{N}-1}x_{k+i}^TQx_{k+i} + x_{k+\hat{N}}^TPx_{k+\hat{N}}$.

\textit{Rewriting $\mathit{\hat{J}^{1}_k}$ and $\mathit{\hat{J}^{2}_k}$:} \label{sec_rewriting_J1_J2}
Due to \eqref{eq_uki_extended_state_2} and due to \eqref{eq_xki_extended_state_2}, 
\begin{align}
	&\hat{J}^{1}_k = 
	\hat{u}_k^T\hat{R}_1(\tilde{H}_k,\vec{D}_k)\hat{u}_k + 2\hat{u}_k^T\hat{H}_1(\tilde{H}_k,\vec{D}_k)\hat{x}_k
	+ \text{terms constant w.r.t. $\hat{u}_k$}\label{Jk1_rewritten}
\end{align}
with $\hat{R}_1(\tilde{H}_k,\vec{D}_k) = \sum\nolimits_{i=\underline{d}}^{N-1} \hat{B}_u(i,\tilde{H}_k,\vec{D}_k)^T R \hat{B}_u(i,\tilde{H}_k,\vec{D}_k)$ and\\
$\hat{H}_1(\tilde{H}_k,\vec{D}_k) = \sum\nolimits_{i=\underline{d}}^{N-1} \hat{B}_u(i,\tilde{H}_k,\vec{D}_k)^T R\hat{A}_u(i,\tilde{H}_k,\vec{D}_k)$
and
\begin{align}
	&\hat{J}^{2}_k = 
	\hat{u}_k^T\hat{R}_2(\tilde{H}_k,\vec{D}_k)\hat{u}_k + 2\hat{u}_k^T\hat{H}_2(\tilde{H}_k,\vec{D}_k)\hat{x}_k
	+ \text{terms constant w.r.t. $\hat{u}_k$}\label{Jk2_rewritten}
\end{align}
with $\hat{R}_2(\tilde{H}_k,\vec{D}_k) = \sum\nolimits_{i=\underline{d}+1}^{N} \hat{B}(i,\tilde{H}_k,\vec{D}_k)^T Q \hat{B}(i,\tilde{H}_k,\vec{D}_k)\hat{u}_k$ and\\
$\hat{H}_2(\tilde{H}_k,\vec{D}_k) = \sum\nolimits_{i=\underline{d}+1}^{N} \hat{B}(i,\tilde{H}_k,\vec{D}_k)^T Q \hat{A}(i,\tilde{H}_k,\vec{D}_k)$.\\

\textit{Rewriting $\mathit{u_{k+i}}$ and $\mathit{x_{k+i}}$ for $\mathit{i \geq N}$:}
$\mathbb{E}\left\{\left. u_{k+i} \right|\mathcal{I}_{k}\right\}=
\sum\nolimits_{\delta \in \mathcal{D}_k(\tilde{\mathcal{D}}_k)}  P_k(\delta)\\ \big(\vec{A}_u(i,\delta,\tilde{H}_k,\tilde{\mathcal{D}}_k)\hat{x}_k 
+ \vec{B}_u(i,\delta,\tilde{H}_k,\tilde{\mathcal{D}}_k)\hat{u}_k\big)$ for $-\tilde{H}_k \leq i \leq N-1$ due to \eqref{eq_uki_extended_state_2} and \eqref{eq_expected_values} with
$\vec{A}_u(i,\delta,\tilde{H}_k,\tilde{\mathcal{D}}_k) = \sum\nolimits_{\vec{d}_k \in \hat{\mathcal{D}}_k(\tilde{\mathcal{D}}_k, \delta)}  \vec{P}(\tilde{H}_k, \vec{d}_k) \hat{A}_u(i,\tilde{H}_k,\vec{d}_k)$ and\\
$\vec{B}_u(i,\delta,\tilde{H}_k,\tilde{\mathcal{D}}_k) = \sum\nolimits_{\vec{d}_k \in \hat{\mathcal{D}}_k(\tilde{\mathcal{D}}_k, \delta)}  \vec{P}(\tilde{H}_k, \vec{d}_k) \hat{B}_u(i,\tilde{H}_k,\vec{d}_k)$. Therefore, \eqref{eq_kappa_ki_proposed} can be written as
\begin{align}
	&\kappa_{k}^{(i)} = -Lx_{k+i} + \!\!\!\! \sum\limits_{\delta \in \mathcal{D}_k(\tilde{\mathcal{D}}_k)} \!\!\!\!  P_k(\delta) 
	\big(\vec{L}_x(i,\delta,\tilde{H}_k,\tilde{\mathcal{D}}_k,\vec{D}_k)\hat{x}_k 
	+ \vec{L}_u(i,\delta,\tilde{H}_k,\tilde{\mathcal{D}}_k,\vec{D}_k)\hat{u}_k\big)\label{eq_kappa_ki_rewritten_1}
\end{align}
using $\sum_{\delta \in \mathcal{D}_k(\tilde{\mathcal{D}}_k)}  P_k(\delta) = 1$ with $\vec{L}_x(i,\delta,\tilde{H}_k,\tilde{\mathcal{D}}_k,\vec{D}_k) = 
L\sum\nolimits_{j = -\hat{H}_k^{(i)}}^{\overline{d}-1} A^{i-1-j}B\\
\big(\hat{A}_u(j,\tilde{H}_k,\vec{D}_k)-\vec{A}_u(j,\delta,\tilde{H}_k,\tilde{\mathcal{D}}_k)\big)$ and $\vec{L}_u(i,\delta,\tilde{H}_k,\tilde{\mathcal{D}}_k,\vec{D}_k) = 
L\sum\nolimits_{j = -\hat{H}_k^{(i)}}^{\overline{d}-1}\\ A^{i-1-j}B
\big(\hat{B}_u(j,\tilde{H}_k,\vec{D}_k)
- \vec{B}_u(j,\delta,\tilde{H}_k,\tilde{\mathcal{D}}_k)\big)$.

For $i\geq N$, $x_{k+i}$ and $u_{k+i}$ can be written in the form
\begin{align}
	x_{k+i} &= (A-BL)^{i-N}x_{k+N} + \Delta x_{k+i}\label{x_ki_rewritten_i_geq_N}
\end{align}
and $u_{k+i} = -Lx_{k+i} + \Delta u_{k+i}$ with $\Delta x_{k+N} = 0$ so
$x_{k+i+1} = Ax_{k+i}+Bu_{k+i} = (A-BL)^{i+1-N}x_{k+N} + \Delta x_{k+i+1}$ with $\Delta x_{k+i+1} = (A-BL)\Delta x_{k+i} + B\Delta u_{k+i}$. Therefore, due to \eqref{eq_kappa_ki_rewritten_1}, $\Delta x_{k+i} = \sum\nolimits_{j=N}^{i-1}(A-BL)^{i-1-j}B\Delta u_{k+j}$ and
$\Delta u_{k+i} = \sum\nolimits_{\delta \in \mathcal{D}_k(\tilde{\mathcal{D}}_k)} P_k(\delta) \big(\vec{L}_x(i,\delta,\tilde{H}_k,\tilde{\mathcal{D}}_k,\vec{D}_k)\hat{x}_k 
+ \vec{L}_u(i,\delta,\tilde{H}_k,\tilde{\mathcal{D}}_k,\vec{D}_k)\hat{u}_k\big)$. Therefore, \eqref{x_ki_rewritten_i_geq_N} can be written as
\begin{align}
	&x_{k+i} = \label{eq_xki_i_geq_N}
	\sum\limits_{\delta \in \mathcal{D}_k(\tilde{\mathcal{D}}_k)}  P_k(\delta)\big(\hat{\hat{A}}(i,\delta,\tilde{H}_k,\tilde{\mathcal{D}}_k,\vec{D}_k)\hat{x}_k + \hat{\hat{B}}(i,\delta,\tilde{H}_k,\tilde{\mathcal{D}}_k,\vec{D}_k)\hat{u}_k\big)
\end{align}
with
$\hat{\hat{A}}(i,\delta,\tilde{H}_k,\tilde{\mathcal{D}}_k,\vec{D}_k)
= (A-BL)^{i-N}\hat{A}(N,\tilde{H}_k,\vec{D}_k) 
+ \sum\nolimits_{j=N}^{i-1}(A-BL)^{i-1-j}B\vec{L}_x(j,\delta,\tilde{H}_k,\tilde{\mathcal{D}}_k,\vec{D}_k)$ and
$\hat{\hat{B}}(i,\delta,\tilde{H}_k,\tilde{\mathcal{D}}_k,\vec{D}_k)
= (A-BL)^{i-N}\hat{B}(N,\tilde{H}_k,\vec{D}_k)
+ \sum\nolimits_{j=N}^{i-1}(A-BL)^{i-1-j}B \vec{L}_u(j,\delta,\tilde{H}_k,\tilde{\mathcal{D}}_k,\vec{D}_k)$ using \eqref{eq_xki_extended_state_2}.

Inserting \eqref{eq_xki_i_geq_N} in \eqref{eq_kappa_ki_rewritten_1}, $u_{k+i} = \kappa_{k}^{(i)}$ can be written as
\begin{align}
	&u_{k+i} = \label{eq_uki_i_geq_N}
	\!\!\! \sum\limits_{\delta \in \mathcal{D}_k(\tilde{\mathcal{D}}_k)}  \!\!\! P_k(\delta)\big(\hat{\hat{A}}_u(i,\delta,\tilde{H}_k,\tilde{\mathcal{D}}_k,\vec{D}_k)\hat{x}_k 
	+ 
	\hat{\hat{B}}_u(i,\delta,\tilde{H}_k,\tilde{\mathcal{D}}_k,\vec{D}_k)\hat{u}_k\big),
	\ \ i \geq N
\end{align}
with $\hat{\hat{A}}_u(i,\delta,\tilde{H}_k,\tilde{\mathcal{D}}_k,\vec{D}_k)
= -L\hat{\hat{A}}(i,\delta,\tilde{H}_k,\tilde{\mathcal{D}}_k,\vec{D}_k) +\vec{L}_x(i,\delta,\tilde{H}_k,\tilde{\mathcal{D}}_k,\vec{D}_k)$ and 
$\hat{\hat{B}}_u(i,\delta,\tilde{H}_k,\tilde{\mathcal{D}}_k,\vec{D}_k)
= -L\hat{\hat{B}}(i,\delta,\tilde{H}_k,\tilde{\mathcal{D}}_k,\vec{D}_k)  +\vec{L}_u(i,\delta,\tilde{H}_k,\tilde{\mathcal{D}}_k,\vec{D}_k)$.


\textit{Rewriting $\mathit{\hat{J}^{3}_k}$ and $\mathit{\hat{J}^{4}_k}$:}
Due to \eqref{eq_uki_i_geq_N} and due to \eqref{eq_xki_i_geq_N} 
\begin{align}
	&\hat{J}^{3}_k = \label{Jk3_rewritten}
	\sum\nolimits_{\delta_1 \in \mathcal{D}_k(\tilde{\mathcal{D}}_k)} \sum\nolimits_{\delta_2 \in \mathcal{D}_k(\tilde{\mathcal{D}}_k)} P_k(\delta_1)P_k(\delta_2)
	\big(\hat{u}_k^T\hat{R}_3(\delta_1,\delta_2,\tilde{H}_k,\tilde{\mathcal{D}}_k,\vec{D}_k)\hat{u}_k
	\\[-0.5mm]
	&\ \ + 2\hat{u}_k^T\hat{H}_3(\delta_1,\delta_2,\tilde{H}_k,\tilde{\mathcal{D}}_k,\vec{D}_k)\hat{x}_k\big)+ \text{terms constant w.r.t. $\hat{u}_k$ with}\nonumber
\end{align}
$\hat{R}_3(\delta_1,\delta_2,\tilde{H}_k,\tilde{\mathcal{D}}_k,\vec{D}_k) =
\sum\nolimits_{i=N}^{\hat{N}-1} \hat{\hat{B}}_u(i,\delta_1,\tilde{H}_k,\tilde{\mathcal{D}}_k,\vec{D}_k)^TR
\hat{\hat{B}}_u(i,\delta_2,\tilde{H}_k,\tilde{\mathcal{D}}_k,\vec{D}_k)$, 
$\hat{H}_3(\delta_1,\delta_2,\tilde{H}_k,\tilde{\mathcal{D}}_k,\vec{D}_k) =
\sum\nolimits_{i=N}^{\hat{N}-1} \hat{\hat{B}}_u(i,\delta_1,\tilde{H}_k,\tilde{\mathcal{D}}_k,\vec{D}_k)^TR
\hat{\hat{A}}_u(i,\delta_2,\tilde{H}_k,\tilde{\mathcal{D}}_k,\vec{D}_k)$;
\begin{align}
	&\hat{J}^{4}_k = \label{Jk4_rewritten}
	\sum\nolimits_{\delta_1 \in \mathcal{D}_k(\tilde{\mathcal{D}}_k)} \sum\nolimits_{\delta_2 \in \mathcal{D}_k(\tilde{\mathcal{D}}_k)} P_k(\delta_1)P_k(\delta_2)
	\big(\hat{u}_k^T\hat{R}_4(\delta_1,\delta_2,\tilde{H}_k,\tilde{\mathcal{D}}_k,\vec{D}_k)\hat{u}_k
	\\[-0.5mm]
	&\ \ + 2\hat{u}_k^T\hat{H}_4(\delta_1,\delta_2,\tilde{H}_k,\tilde{\mathcal{D}}_k,\vec{D}_k)\hat{x}_k\big)+ \text{terms constant w.r.t. $\hat{u}_k$ with}\nonumber
\end{align}
$\hat{R}_4(\delta_1,\delta_2,\tilde{H}_k,\tilde{\mathcal{D}}_k,\vec{D}_k) =
\sum\nolimits_{i=N+1}^{\hat{N}-1} \hat{\hat{B}}(i,\delta_1,\tilde{H}_k,\tilde{\mathcal{D}}_k,\vec{D}_k)^TQ
\hat{\hat{B}}(i,\delta_2,\tilde{H}_k,\tilde{\mathcal{D}}_k,\vec{D}_k)  +\\ \hat{\hat{B}}(\hat{N},\delta_1,\tilde{H}_k,\tilde{\mathcal{D}}_k,\vec{D}_k)^TP
\hat{\hat{B}}(\hat{N},\delta_2,\tilde{H}_k,\tilde{\mathcal{D}}_k,\vec{D}_k)$ and $\hat{H}_4(\delta_1,\delta_2,\tilde{H}_k,\tilde{\mathcal{D}}_k,\vec{D}_k) = \\
\sum\nolimits_{i=N+1}^{\hat{N}-1} \hat{\hat{B}}(i,\delta_1,\tilde{H}_k,\tilde{\mathcal{D}}_k,\vec{D}_k)^TQ
\hat{\hat{A}}(i,\delta_2,\tilde{H}_k,\tilde{\mathcal{D}}_k,\vec{D}_k)  + \hat{\hat{B}}(\hat{N},\delta_1,\tilde{H}_k,\tilde{\mathcal{D}}_k,\vec{D}_k)^TP
\hat{\hat{A}}(\hat{N},\delta_2,\tilde{H}_k,\tilde{\mathcal{D}}_k,\vec{D}_k)$.

\begin{lemma}[Rewritten Cost]\label{le_rewritten_cost}
	The cost function in \eqref{eq_cost_stochastic} can be written as terms constant w.r.t. $\hat{u}_k\ \ +$
	\begin{align}
		&\sum\nolimits_{\delta_1 \in \tilde{\mathcal{D}}_k} \sum\nolimits_{\delta_2 \in \tilde{\mathcal{D}}_k}\sum\nolimits_{\delta_3 \in \tilde{\mathcal{D}}_k} P_k(\delta_1)P_k(\delta_2)P_k(\delta_3)\cdot \label{eq_form_J_hat}\\
		&\big(\hat{u}_k^T\hat{R}(\delta_1,\delta_2,\delta_3,\tilde{H}_k,\tilde{\mathcal{D}}_k)\hat{u}_k + 2\hat{u}_k^T\hat{H}(\delta_1,\delta_2,\delta_3,\tilde{H}_k,\tilde{\mathcal{D}}_k)\hat{x}_k\big)\nonumber
	\end{align}
	with $\hat{R}(\delta_1,\delta_2,\delta_3,\tilde{H}_k,\tilde{\mathcal{D}}_k) = 
	\sum\nolimits_{\vec{d}_k \in \hat{\mathcal{D}}_k(\tilde{\mathcal{D}}_k, \delta_3)} \vec{P}(\tilde{H}_k, \vec{d}_k)\big(\hat{R}_1(\vec{d}_k)+
	\hat{R}_2(\tilde{H}_k,\vec{d}_k) + \hat{R}_3(\delta_1,\delta_2,\tilde{H}_k,\tilde{\mathcal{D}}_k,\vec{d}_k)
	+ \hat{R}_4(\delta_1,\delta_2,\tilde{H}_k,\tilde{\mathcal{D}}_k,\vec{d}_k)\big)$ and
	$\hat{H}(\delta_1,\delta_2,\delta_3,\tilde{H}_k,\tilde{\mathcal{D}}_k) = \\
	\sum\nolimits_{\vec{d}_k \in \hat{\mathcal{D}}_k(\tilde{\mathcal{D}}_k, \delta_3)} \vec{P}(\tilde{H}_k, \vec{d}_k)\big(\hat{H}_1(\vec{d}_k)+
	\hat{H}_2(\tilde{H}_k,\vec{d}_k) + \hat{H}_3(\delta_1,\delta_2,\tilde{H}_k,\tilde{\mathcal{D}}_k,\vec{d}_k)
	+ \hat{H}_4(\delta_1,\delta_2,\tilde{H}_k,\tilde{\mathcal{D}}_k,\vec{d}_k)\big)$.
\end{lemma}
\begin{proof}
	Using \eqref{eq_Jk_hat}, the cost function in \eqref{eq_cost_stochastic} can be written as $\hat{J}_k +$ terms constant w.r.t. $\hat{u}_k$ where $\hat{J}_k$ can be written in the form \eqref{eq_form_J_hat} by inserting \eqref{Jk1_rewritten}, \eqref{Jk2_rewritten}, \eqref{Jk3_rewritten} and \eqref{Jk4_rewritten} in \eqref{eq_Jk_hat} and rewriting the expected value via \eqref{eq_expected_values}.
\end{proof}

\subsubsection{Optimization Problem}
Due to Lemma~\ref{le_all_constraints} and Lemma~\ref{le_rewritten_cost},
\begin{align}
	&\underset{\hat{u}_{k}}{min} \left(\hat{u}_k^T\vec{R}_k\hat{u}_k + 2\hat{u}_k^T\vec{H}_k\hat{x}_k\right)\label{eq_opt_pr_final}\\
	&s.t.\ \ \vec{M}_{x}(\tilde{H}_k,\tilde{\mathcal{D}}_k)\hat{x}_k + \vec{M}_{u}(\tilde{H}_k,\tilde{\mathcal{D}}_k)\hat{u}_k \leq \vec{n}(\tilde{\mathcal{D}}_k) \quad \text{with}\nonumber
\end{align}
$\vec{R}_k = \sum\limits_{\delta_1 \in \mathcal{D}_k(\tilde{\mathcal{D}}_k)}\! \sum\limits_{\delta_2 \in \mathcal{D}_k(\tilde{\mathcal{D}}_k)}\!\sum\limits_{\delta_3 \in \mathcal{D}_k(\tilde{\mathcal{D}}_k)}\!P_k(\delta_1)P_k(\delta_2)$ $ P_k(\delta_3)\hat{R}(\delta_1,\delta_2,\delta_3,\tilde{H}_k,\tilde{\mathcal{D}}_k)$ and
$\vec{H}_k = \sum\limits_{\delta_1 \in \mathcal{D}_k(\tilde{\mathcal{D}}_k)} \sum\limits_{\delta_2 \in \mathcal{D}_k(\tilde{\mathcal{D}}_k)} $ $ \sum\limits_{\delta_3 \in \mathcal{D}_k(\tilde{\mathcal{D}}_k)}P_k(\delta_1)P_k(\delta_2)P_k(\delta_3)\hat{H}(\delta_1,\delta_2,\delta_3,\tilde{H}_k,\tilde{\mathcal{D}}_k)$
is\\ equal to \eqref{eq_optimization_problem} with $\kappa_k^{(i)}$ and $\mathcal{X}_N(\mathcal{I}_k)$ from Definition~\ref{def_stabilizing}~and~\ref{def_terminal} except for constraints that do not depend on $\hat{u}_{k}$ and terms in the cost function that are constant with respect to $\hat{u}_{k}$.

The matrices $\vec{M}_{x}(\tilde{H}_k,\tilde{\mathcal{D}}_k)$, $\vec{M}_{u}(\tilde{H}_k,\tilde{\mathcal{D}}_k)$, $\vec{n}(\tilde{\mathcal{D}}_k)$, $\hat{R}(\delta_1,\delta_2,\delta_3,\tilde{H}_k,\tilde{\mathcal{D}}_k)$ and $\hat{H}(\delta_1,\delta_2,\delta_3,\tilde{H}_k,\tilde{\mathcal{D}}_k)$ can be computed offline in advance where $\tilde{\mathcal{D}}_k \subseteq [\underline{d},\, \overline{d}]$,  $\delta_1,\delta_2,\delta_3 \in [\underline{d},\, \overline{d}]$ and $\tilde{H}_k \in [\min(\underline{h}, \underline{s}-1),\, \overline{h}]$. 

The optimization problem \eqref{eq_opt_pr_final} can also be written in the form \eqref{eq_final_opt_pr_form} which completes the proof for Theorem~\ref{th_1}.

\subsubsection{Remarks}
The computational effort required to solve \eqref{eq_opt_pr_final} primarily depends on the number of optimization variables $\hat{m} = m(N-\underline{d})$ where $N>\overline{d}$. This is comparable to a standard MPC for a setup without networks which requires $mN$ optimization variables where $N>0$. However, the number of linear inequalities can be much larger.

To reduce number of inequalities in \eqref{eq_constraints_redundant},
$(a)$ include that, for admissible $x_0$, $\hat{x}_k$ certainly satisfies $M_xx_{k-\tilde{H}_k}\leq n_x$ and $M_u\tilde{u}_{k-i}^{(d)} \leq n_u$ for all $i\in[1,\overline{d}+\overline{h}]$, $d\in[\underline{d},\overline{d}]$ and rewrite the resulting constraints in the form $\vec{M}_{y}(\tilde{H}_k,\tilde{\mathcal{D}}_k)\hat{y}_k \leq \vec{n}_y(\tilde{\mathcal{D}}_k)$;
$(b)$ remove redundant constraints using the MPT-Toolbox \cite{Herceg13};
$(c)$ rewrite the resulting constraints in the form \eqref{eq_all_constr_Dk_vec} and remove constraints that do not depend on $\hat{u}_k$.

The matrices in \eqref{eq_opt_pr_final} do not have to be computed for all $\tilde{\mathcal{D}}_k \subseteq [\underline{d},\, \overline{d}]$. Sets $\tilde{\mathcal{D}}_k$ that can occur can be determined using that $\mathcal{I}_k$ does not contain data that depends on $(D_{k+i})_{i \geq -\tilde{H}_k}$.

\section{Simulation Example}
The plant is the controllable third order LTI system
\begin{align*}
	x_{k+1} &= Ax_k+Bu_k = \begin{bmatrix}
		0.8& 0.5& 0\\ 0 &-1.2& 0.2\\ 0& 0& 0.2
	\end{bmatrix}x_k + \begin{bmatrix}
		1& 0\\ 0& 0\\ 0& 1
	\end{bmatrix}u_k.
\end{align*}

The bounds for the random variables $D_k$, $H_k$ and $S_k$ are $\underline{d} = 0$, $\overline{d} = 2$, $\underline{h} = 0$, $\overline{h} = 1$, $\underline{s} = 1$, $\overline{s} = 3$
and the parameters of the homogeneous Markov Process $\mathbf{D}$ are given by
$\begin{bmatrix}
	\mu(0)& \mu(1)& \mu(2)
\end{bmatrix} = \begin{bmatrix}
	0.2& 0.4& 0.4
\end{bmatrix}$ and
\begin{align*}
	\begin{bmatrix}
		\Phi(0,0) &\Phi(0,1) &\Phi(0,2)\\
		\Phi(1,0) &\Phi(1,1) &\Phi(1,2)\\
		\Phi(2,0) &\Phi(2,1) &\Phi(2,2)
	\end{bmatrix}
	= 
	\begin{bmatrix}
		0.4 &0.6 &0\\
		0.4 &0.4 &0.2\\
		0.2 &0.4 &0.4
	\end{bmatrix}.
\end{align*}
The constraints on state and actuating variables are
\begin{align*}
	\begin{bmatrix}
		-10\\ -5\\ -10
	\end{bmatrix} \leq \begin{bmatrix}
		x_{1,k}\\ x_{2,k}\\ x_{3,k}
	\end{bmatrix} \leq \begin{bmatrix}
		10\\ 5\\ 10
	\end{bmatrix}, \quad
	\begin{bmatrix}
		-2\\ -5
	\end{bmatrix} \leq \begin{bmatrix}
		u_{1,k}\\ u_{2,k}
	\end{bmatrix} \leq \begin{bmatrix}
		2\\ 5
	\end{bmatrix} \quad \forall\, k\geq 0
\end{align*}
which can be written in the form \eqref{constraints}. The controller parameters are given by $N = 10$, $Q = \text{diag}(10,100,1)$ and $R = \text{diag}(1,1)$. The resulting set of admissible $x_0$ is depicted in Fig.~\ref{fig2}. All simulations are executed with $x_0 = \begin{bmatrix}
	-4.5& -2.6& -7
\end{bmatrix}^T$ which is admissible.
\begin{figure}[!t]
	\centerline{\includegraphics[width=0.5\columnwidth]{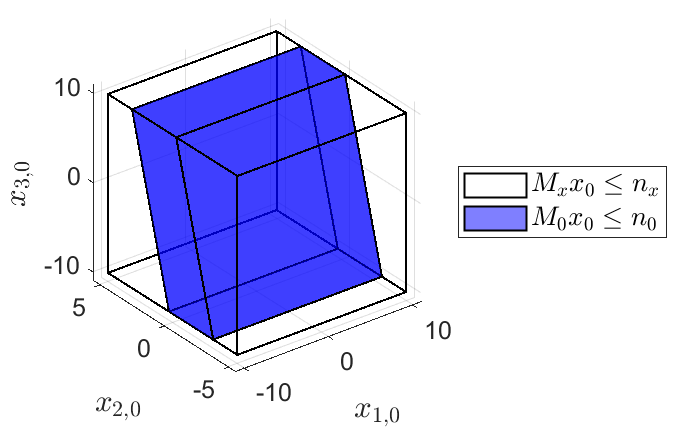}}
	\caption{Set of admissible initial states: $x_0$ satisfying $M_0x_0 \leq n_0$.}
	\label{fig2}
\end{figure}

The proposed control strategy is applied in a simulation with randomly generated realizations of $\mathbf{D}$, $\mathbf{H}$ and $\mathbf{S}$. These realizations and the resulting evolution of $x_k$ and $u_k$ are depicted in Fig.~\ref{fig3} which illustrates that applying the proposed control strategy satisfies all constraints on state an actuating variables while $x_k$ converges to zero with probability one.

The performance of the proposed stochastic MPC is compared to a deterministic MPC which is obtained by $(a)$ using a buffer to ensure that $D_k = \overline{d}$ for all $k\geq 0$ and $(b)$ solving \eqref{eq_optimization_problem} with $\kappa_k^{(i)} = -Lx_{k+i}$ and the terminal constraint $M_Nx_{k+N}\leq n_N$. These control strategies are compared by comparing the evolution of $\tilde{J}_k = \sum\nolimits_{i=0}^{k} \left(x_i^TQx_i + u_i^TRu_i\right)$.

Three simulations are executed with $H_k = \overline{h}\ \forall\, k$:
$(1)$ the stochastic MPC with $D_k = \overline{d}\ \forall\, k$ and $S_k = \overline{s}\ \forall\, k$;
$(2)$ the stochastic MPC with $D_k = \underline{d}\ \forall\, k$ and $S_k = \underline{s}\ \forall\, k$;
$(3)$ the deterministic MPC (which yields the same results for all possible realizations of  $\mathbf{D}$ and $\mathbf{S}$).
The resulting evolution of $\tilde{J}_k$ is depicted in Fig.~\ref{fig4} which illustrates the potential increase in performance achieved by minimizing the expected value of the cost function: $(i)$ the deterministic MPC is designed by minimizing $\tilde{J}_{\infty}$ for $D_k = \overline{d}\ \forall\, k$. Since this is not the only possible case, applying the stochastic MPC yields a larger cost if $D_k = \overline{d}\ \forall\, k$. However, this difference is quite small. $(ii)$ if $D_k = \underline{d}\ \forall\, k$ and $S_k = \underline{s}\ \forall\, k$, applying the proposed stochastic MPC results in a significantly lower value of the cost function than applying the deterministic MPC. 

The simulation with maximal delays is repeated with (a) an LQR controller such that $u_k = 0\ \forall\, k < \overline{h}$, $u_k = -Lx_k\ \forall\, k \geq \overline{h}$ and (b) the stochastic MPC without the constraints which both results in constraint violations.

\begin{figure}[!t]
	\centerline{\includegraphics[width=0.85\columnwidth]{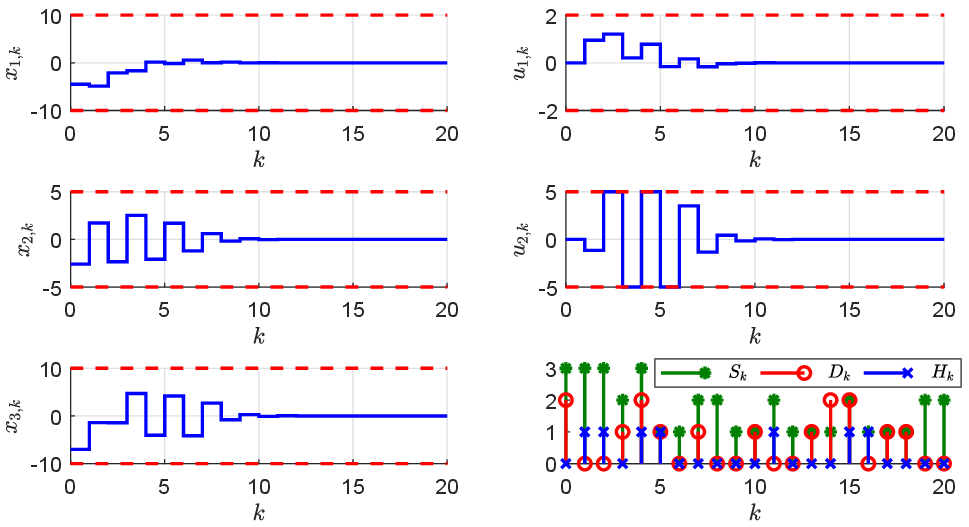}}
	\caption{Simulation with randomly generated realizations of $\mathbf{D}$, $\mathbf{H}$ and $\mathbf{S}$.}
	\label{fig3}
\end{figure}
\begin{figure}[!t]
	\centerline{\includegraphics[width=0.85\columnwidth]{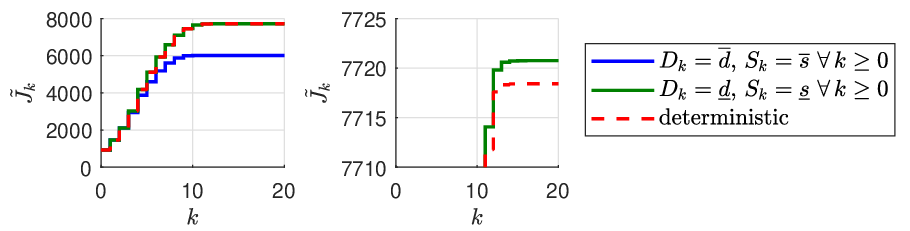}}
	\caption{Comparing stochastic and deterministic MPC for $H_k = \overline{h}\ \forall\, k$.}
	\label{fig4}
\end{figure}
\section{Conclusion and Further Work}
Compared to the MPC minimizing the value of the cost function for $D_k = \overline{d}\ \forall\, k$, the proposed stochastic MPC $(i)$ potentially increases the performance significantly $(ii)$ without deteriorating the performance for $D_k = \overline{d}\ \forall\, k$ significantly.

However, the following issue has to be handled when considering real-world applications: for example, if $x_{k-1}$ and $x_{k}$ are known but $D_{k-1}$ is not, computing probabilities like in Section~\ref{sec_implement_probabilities} utilizes the set of all $d_{k-1}$ for which $x_{k} = Ax_{k-1} + B\tilde{u}_{k-1-d_{k-1}}^{(d_{k-1})}$. This set will typically be empty when considering real-world applications where there is a mismatch between the plant and its model. While this can be handled by, for example, computing the set of all $d_{k-1}$ for which some norm of the difference between $x_{k}$ and $Ax_{k-1} + B\tilde{u}_{k-1-d_{k-1}}^{(d_{k-1})}$ is sufficiently small, including a suitable model of disturbances is considered to be the most important next step.

Additionally, the properties of the stochastic processes might not be known exactly in real-world applications. Deviating transition probabilities do not seem to be problematic like, for example, in the simulation where $D_k = \overline{d}\ \forall\, k$ is enforced. However, constraint violations might occur and feasibility is no longer guaranteed if, e.g., $D_k = d_k > \overline{d}$ occurs.

\newpage
\appendices
\section{}\label{Appendix_compute_x_Hk_hat}
Due to \eqref{eq_uk}, the sequence of actuating variables given by \eqref{eq_uki_stochastic} can be actually applied to the plant if it is possible to choose the controller output $\tilde{u}_{k+i}$ as
\begin{align}
	\tilde{u}_{k+i}^{(d)} = \begin{cases}
		\hat{u}_k^{(i+d)}& i + d \leq N-1\\
		\kappa_{k}^{(i+d)} & N \leq i+d,
	\end{cases} \qquad \  d \in[\underline{d},\overline{d}]\label{eq_5}
\end{align}
for all $i \geq 0$. This is possible if $\kappa_{k}^{(N)}$ can be computed at at time step $k+N-\overline{d}$, if   $\kappa_{k}^{(N+1)}$ can be computed at time step $k+N-\overline{d}+1$ and so on, i.e. it has to be possible to compute $\kappa_{k}^{(j)}$ at time step $k+j-\overline{d}$ for all $j \geq N$.

If $\kappa_{k}^{(l)}$ can be computed at time step $k+l-\overline{d}$ for all $l \in [N,j-1]$, which certainly holds for $j=N$, then $\tilde{u}_{k+i}^{(d)}$ can be computed via \eqref{eq_5} for all $i\geq 0,\, d \in[\underline{d},\overline{d}]$ for which $i+d \leq j-1$. This results in
\begin{align}\label{eq_20}
	u_{k+l} = \tilde{u}_{k+l-D_{k+l}}^{(D_{k+l})} = \begin{cases}
		\hat{u}_k^{(l)}& D_{k+l} \leq l \leq N-1\\
		\kappa_{k}^{(l)}& N \leq l \leq j-1
	\end{cases}
\end{align}
where $\kappa_{k}^{(l)}$ is known at time step $k+j-\overline{d}$ for all $l\in[N,j-1]$. It remains to be shown that this implies that it is also possible to compute $\kappa_{k}^{(j)}$ at time step $k+j-\overline{d}$ in order to show that it is possible to determine $\tilde{u}_{k+i}$ via \eqref{eq_5} for all $i \geq 0$.

Due to \eqref{eq_hk_hat}, $\hat{H}_k^{(j)}+j = \min(\overline{h}+N-1, \overline{s}+N-2,\tilde{H}_k+j)$.
Since $\overline{h}\geq 0$ and $\overline{s} \geq 0$ and since $\tilde{H}_k \geq -1$ due to \eqref{eq_hk_tilde}, $\hat{H}_k^{(j)}+j = \min(\overline{h}+N-1, \overline{s}+N-2)$ for all $j \geq N$.
Since $N > \overline{d} \geq 1$, this yields $\hat{H}_k^{(j)}+j \geq 0$ for all $j \geq N$ so $x_{k+j} = A^{\hat{H}_k^{(j)}+j}x_{k-\hat{H}_k^{(j)}} + \sum\nolimits_{l = -\hat{H}_k^{(j)}}^{j-1} A^{j-1-l}Bu_{k+l}$. Inserting this result in \eqref{eq_kappa_ki_proposed} yields
\begin{align*}\kappa_{k}^{(j)} = -LA^{\hat{H}_k^{(j)}+j}x_{k-\hat{H}_k^{(j)}} - L \sum\nolimits_{l = \overline{d}}^{j-1} A^{j-1-l}Bu_{k+l}
- L \sum\nolimits_{l = -\hat{H}_k^{(j)}}^{\overline{d}-1} A^{j-1-l}B\mathbb{E}\left\{\left.u_{k+l}\right|\mathcal{I}_{k}\right\}.
\end{align*}

Due to \eqref{eq_20}, $(u_{k+l})_{\overline{d}\leq l \leq j-1}$ is known at time step $k+j-\overline{d}$; since $k+j-\overline{d} > k$ for $j \geq N$, $\mathcal{I}_{k}$ is available as well. 

$x_{k-\hat{H}_k^{(j)}}$ with $\hat{H}_k^{(j)}$ given by \eqref{eq_hk_hat} can be computed at time step $k+j-\overline{d}$ with $j \geq N$ as well since $(i)$ $x_{k-\tilde{H}_k}$ is known at time step $k < k+j-\overline{d}$ since $x_{k-H_k}$ is known and, if $k \geq K_s$, $\left(u_{\tilde{k}}\right)_{0\leq \tilde{k} \leq k-S_k}$ is also known, $(ii)$ $x_{k-\hat{H}_k^{(j)}}$ is known at time step $k+j-N+1$ if $\hat{H}_k^{(j)} = \overline{h}+N-j-1$, $(iii)$ $\left(u_{k+i}\right)_{-\tilde{H}_k\leq i \leq -(\hat{H}_k^{(j)}+1)}$ is known at time step $k+j-N+1$ if $\hat{H}_k^{(j)} = \overline{s}+N-j-2$ since then $k+j-N+1-S_{k+j-N+1} \geq k-(\hat{H}_k^{(j)}+1)$ and $(iv)$ if $x_{k-\hat{H}_k^{(j)}}$ can be computed at time step $k+j-N+1$ then it is also known at time step $k+j-\overline{d}$.

Therefore, $\kappa_{k}^{(j)}$ can be computed at time step $k+j-\overline{d}$ and therefore, \eqref{eq_uki_stochastic} can actually be applied to the plant for \eqref{stabilizing_control_law}.

\section{}\label{Appendix_eq_kappa_k+1_N-1}
In order to show that $\kappa_{k+1}^{(N-1)} = -L \mathbb{E}\left\{\left.x_{k+N}\right|\mathcal{I}_{k+1}\right\}$ holds for \eqref{eq_uki+1}, $x_{k+N}$ is written as $x_{k+N} = A^{N+\hat{H}_{k+1}^{(N-1)}-1}x_{k+1-\hat{H}_{k+1}^{(N-1)}} + \sum\nolimits_{j = 1-\hat{H}_{k+1}^{(N-1)}}^{N-1} A^{N-1-j}Bu_{k+j}$
where $\hat{H}_{k+1}^{(N-1)} = \tilde{H}_{k+1}$ due to \eqref{eq_hk_hat} so $x_{k+1-\hat{H}_{k+1}^{(N-1)}}$ can be computed from $\mathcal{I}_{k+1}$. For $j \in [\overline{d}+1, N-1]$, $\mathbb{E}\left\{\left.u_{k+j}\right|\mathcal{I}_{k+1}\right\} = u_{k+j}$ for \eqref{eq_uki+1} so $\mathbb{E}\left\{\left.x_{k+N}\right|\mathcal{I}_{k+1}\right\} - x_{k+N}  
=\sum\nolimits_{j = 1-\hat{H}_{k+1}^{(N-1)}}^{\overline{d}} A^{N-1-j}B\left(\mathbb{E}\left\{\left.u_{k+j}\right|\mathcal{I}_{k+1}\right\} - u_{k+j}\right)$
and therefore $x_{k+N} = \mathbb{E}\left\{\left.x_{k+N}\right|\mathcal{I}_{k+1}\right\} - \sum\nolimits_{j = \hat{H}_{k+1}^{(N-1)}}^{\overline{d}-1}  A^{N-2-j}B\left(\mathbb{E}\left\{\left.u_{k+1+j}\right|\mathcal{I}_{k+1}\right\} - u_{k+1+j}\right)$.
Computing $\kappa_{k+1}^{(N-1)}$ via \eqref{eq_kappa_ki_proposed} and inserting $x_{k+N}$  determined above yields $\kappa_{k+1}^{(N-1)} = -L \mathbb{E}\left\{\left.x_{k+N}\right|\mathcal{I}_{k+1}\right\}$.

\section{}\label{Appendix_eq_kappa_k+1_of_dk_vec}
In this Section it is shown that, for \eqref{uki_k+1_rewritten},  \eqref{eq_kappa_k+1_of_dk+1_vec} yields \eqref{eq_kappa_k+1_of_dk_vec} with $\vec{d}_{k} \in \mathcal{D}_k$ for any $\vec{d}_{k+1} \in \mathcal{D}_{k+1}$. According to \eqref{eq_hk_hat}, $\hat{H}_k^{(i)} = \min(\overline{h}+N-i-1, \overline{s}+N-i-2,\tilde{H}_k)$ and $\hat{H}_{k+1}^{(i-1)}-1 = \min(\overline{h}+N-i, \overline{s}+N-i-1,\tilde{H}_{k+1})-1 = \min(\overline{h}+N-i-1, \overline{s}+N-i-2,\tilde{H}_{k+1}-1)$. Since $\tilde{H}_{k+1} \leq \tilde{H}_k +1$, i.e. since $\tilde{H}_{k+1}-1 \leq \tilde{H}_k$, this implies that $(a)$ $\hat{H}_{k+1}^{(i-1)}-1 \leq \hat{H}_k^{(i)}$ always holds and that $(b)$ if $\hat{H}_{k+1}^{(i-1)}-1 < \tilde{H}_{k+1}-1$ then also $\hat{H}_{k}^{(i)} < \tilde{H}_{k}$ which results in $\hat{H}_{k+1}^{(i-1)}-1 = \min(\overline{h}+N-i-1, \overline{s}+N-i-2) = \hat{H}_{k}^{(i)}$. This can be written as
\begin{subequations}
	\begin{align}
		&\hat{H}_{k+1}^{(i-1)}-1 < \tilde{H}_{k+1}-1\ \Rightarrow\ \hat{H}_{k+1}^{(i-1)}-1 = \hat{H}_k^{(i)}\label{eq_a}\\
		& \hat{H}_{k+1}^{(i-1)}-1 = \tilde{H}_{k+1}-1\ \Rightarrow\ \hat{H}_{k+1}^{(i-1)}-1 \leq \hat{H}_k^{(i)}.\label{eq_b}
	\end{align}
\end{subequations}

Due to \eqref{eq_a}, \eqref{eq_kappa_k+1_of_dk+1_vec} with $\hat{H}_{k+1}^{(i-1)} < \tilde{H}_{k+1}$ can be written as 
\begin{align} \label{eq_2}
	& \bar{\kappa}_{k+1}^{(i-1)}(\vec{d}_{k+1}) = -Lx_{k+i}
	+ L \sum\nolimits_{j = -\hat{H}_{k}^{(i)}}^{\overline{d}} A^{i-1-j}B\left(\bar{u}_{k}^{(j)}(D_{k+j}) -\bar{u}_{k}^{(j)}(d_{k+j})\right).
\end{align}

Since $(x_{k+j})_{-k \leq j \leq 1-\tilde{H}_{k+1}}$ can be computed from $\mathcal{I}_{k+1}$ and $x_{k+j+1}-Ax_{k+j} = Bu_{k+j}$, $Bu_{k+j}$ can be computed from $\mathcal{I}_{k+1}$ for all $j\in[-k,-\tilde{H}_{k+1}]$. Due to \eqref{eq_b}, this implies that $B\bar{u}_{k}^{(j)}(D_{k+j}) = Bu_{k+j}$ can be computed from $\mathcal{I}_{k+1}$ for all $j\in -\hat{H}_{k}^{(i)} \leq j \leq -\hat{H}_{k+1}^{(i-1)}$ if $\hat{H}_{k+1}^{(i-1)} = \tilde{H}_{k+1}$. Therefore, $\mathbb{P}\left[\left.D_{k+j} \! =\! d_{k+j}\right| \mathcal{I}_{k+1}\right] = 0$ if $B\bar{u}_{k}^{(j)}(d_{k+j}) \neq B\bar{u}_{k}^{(j)}(D_{k+j})$ for all $j\in -\hat{H}_{k}^{(i)} \leq j \leq -\hat{H}_{k+1}^{(i-1)}$. Since $\vec{d}_{k+1} \in \mathcal{D}_{k+1}$ implies that $\mathbb{P}\left[\left.D_{k+j} = d_{k+j}\right| \mathcal{I}_{k+1}\right] > 0$, this yields $B\big(\bar{u}_{k}^{(j)}(D_{k+j}) - \bar{u}_{k}^{(j)}(d_{k+j})\big) = 0$ for $\vec{d}_{k+1} \in \mathcal{D}_{k+1}$ so \eqref{eq_2} also holds for $\hat{H}_{k+1}^{(i-1)} = \tilde{H}_{k+1}$ for all $\vec{d}_{k+1} \in \mathcal{D}_{k+1}$.

Inserting that, for \eqref{uki_k+1_rewritten}, $\bar{u}_{k}^{(\overline{d})}(D_{k+\overline{d}}) = \bar{u}_{k}^{(\overline{d})}(d_{k+\overline{d}})$ for any $d_{k+\overline{d}}\in[\underline{d},\overline{d}]$ in \eqref{eq_2} yields
$\bar{\kappa}_{k+1}^{(i-1)}(\vec{d}_{k+1}) = \bar{\kappa}_{k}^{(i)}(\vec{d}_k)$
according to \eqref{kappa_k_bar} where $\vec{d}_k\in\mathcal{D}_k$ since $(i)$ $\mathcal{D}_k$ is the set of all possible $\vec{D}_k$ given $\mathcal{I}_{k}$ and $(ii)$ if $\vec{D}_k = \vec{d}_k$ is possible given $\mathcal{I}_{k+1}$ then it is also possible given $\mathcal{I}_{k} \subseteq \mathcal{I}_{k+1}$.

\section{}\label{Appendix_opt_pr_stability}
Since $-\hat{H}_k^{(i)} \geq \overline{d}$ for all $i \geq \hat{N}$ due to \eqref{eq_hk_hat}, \eqref{kappa_stability} yields $u_{k+i} = -Lx_{k+i}$ for all $i \geq \hat{N}$. Therefore, minimizing 
$\mathbb{E}\big\{ \sum\nolimits_{\tilde{k}=k+N}^{\infty} \big(x_{\tilde{k}}^TQx_{\tilde{k}} + u_{\tilde{k}}^TRu_{\tilde{k}}\big) \big|\mathcal{I}_{k+\alpha}\big\}$ with ${\alpha \in\mathbb{N}_0}$
with respect to $(v_{k+j})_{-\hat{H}^{(N)}_{k} \leq j \leq \overline{d}-1}$ subject to \eqref{kappa_stability} where ${-\hat{H}^{(N)}_{k} \leq -\hat{H}^{(i)}_{k}}$ for all $i \geq N$ can be written as
\begin{align} 
	&\underset{(v_{k+j})_{-\hat{H}^{(N)}_{k} \leq j \leq \overline{d}-1}}{\min} \mathbb{E}\left\{ \sum\nolimits_{\tilde{k}=N}^{\hat{N}-1} \left(x_{k+i}^TQx_{k+i} + u_{k+i}^TRu_{k+i}\right)
	+ x_{k+\hat{N}}^TPx_{k+\hat{N}} \big|\mathcal{I}_{k+\alpha}\right\}\nonumber\\
	& s.t. \ u_{k+i} = -Lx_{k+i} + L \sum\nolimits_{j = -\hat{H}_k^{(i)}}^{\overline{d}-1}\!\! A^{i-1-j}B\left(u_{k+j} -v_{k+j}\right).\label{eq_21}
\end{align}
With $\check{v}_k = [\begin{matrix}
	v_{k+\overline{d}-1}^T&
	\hdots&
	v_{k-\hat{H}_k^{(N)}}^T
\end{matrix}]^T$, $\check{u}_k = [\begin{matrix}
	u_{k+\overline{d}-1}&
	\hdots&
	u_{k-\hat{H}_k^{(N)}}^T
\end{matrix}]^T$ and $\check{x}_k = [\begin{matrix}
	x_{k+N}^T& \check{u}_k^T
\end{matrix}]^T$; \eqref{eq_21} can be written in the form $\underset{\check{v}_k}{\min}\ \mathbb{E}\big\{ 
\check{x}_k^T\check{Q}\check{x}_k + 2\check{v}_k^T\check{H}\check{x}_k + \check{v}_k^T\check{R}\check{v}_k
\big|\mathcal{I}_{k+\alpha}\big\}$ with $\check{R} = \check{R}^T \succeq 0$. Since $\check{x}_k$ is constant w.r.t. $\check{v}_k$ and since this optimization problem is convex, $\check{v}_k = \check{v}_k^*$ is optimal if
\begin{align}
	\check{R}\check{v}_k^* + \check{H}\mathbb{E}\left\{\left.\check{x}_k\right|\mathcal{I}_{k+\alpha}\right\} = 0. \label{eq_34}
\end{align}
In order to show that $v_{k+j}^* = \mathbb{E}\left\{\left.u_{k+j}\right|\mathcal{I}_{k+\alpha}\right\}$ is optimal, it remains to be shown that \eqref{eq_34} holds for
\begin{align} \label{eq_35}
	\check{v}_k^* &= \mathbb{E}\left\{\left.\check{u}_k\right|\mathcal{I}_{k+\alpha}\right\} = \begin{bmatrix}
		0& I
	\end{bmatrix}\mathbb{E}\left\{\left.\check{x}_k\right|\mathcal{I}_{k+\alpha}\right\}.
\end{align}

This certainly holds if $\alpha$ is large enough such that $\mathbb{E}\left\{\left.\check{x}_k\right|\mathcal{I}_{k+\alpha}\right\} = \check{x}_k$ since $\check{v}_k = \check{u}_k$ results in $u_{k+i} = -Lx_{k+i}$ for all $i\geq N$ which minimizes \eqref{eq_21}. In other words, $
\check{R}\begin{bmatrix}
	0& I
\end{bmatrix}\check{x}_k + \check{H}\check{x}_k = 0$ holds for any possible $\check{x}_k$. 

While $\check{x}_k$ is not known for small values of $\alpha$, $\check{R}\begin{bmatrix}
	0& I
\end{bmatrix}\check{x}_k + \check{H}\check{x}_k = 0$ for all possible $\check{x}_k$ so $\check{R}\begin{bmatrix}
	0& I
\end{bmatrix}\mathbb{E}\left\{\left.\check{x}_k\right|\mathcal{I}_{k+\alpha}\right\} + \check{H}\mathbb{E}\left\{\left.\check{x}_k\right|\mathcal{I}_{k+\alpha}\right\} = 0$
so \eqref{eq_34} holds for \eqref{eq_35} for all ${\alpha \in\mathbb{N}_0}$.

\section{}\label{Appendix_eq_Pk}
\subsubsection{Definitions}
\begin{align}
	&\tilde{P}_1(\underline{i},\overline{i},\delta_{\underline{i}}, (\Theta_i)_{\underline{i}+1 \leq i \leq \overline{i}})\label{Definition_P1_tilde}
	=\mathbb{P}\left[\left. D_{i} \in \Theta_{i} \, \forall\, i \in [\underline{i}+1, \overline{i}] \right| D_{\underline{i}} = \delta_{\underline{i}}
	\right] \quad \forall\, \overline{i} > \underline{i}
\end{align}
which can be written as \eqref{eq_P1_tilde} using \eqref{eq_Markov} and \eqref{eq_Phi}.
$\tilde{P}_2(\underline{i},\overline{i},\delta_{\underline{i}}, \delta_{\overline{i}}, (\Theta_i)_{\underline{i}+1 \leq i \leq \overline{i} - 1}) 
\\ =
\mathbb{P}\left[\left.D_{\overline{i}} = \delta_{\overline{i}} \right|
D_{\underline{i}} = \delta_{\underline{i}},\,
D_{i} \in \Theta_{i} \, \forall\, i \in [\underline{i}+1, \overline{i}-1]\right]$
which can be written as \\
$\frac
{\mathbb{P}\left[\left.D_{\overline{i}} = \delta_{\overline{i}},\,
	D_{i} \in \Theta_{i} \, \forall\, i \in [\underline{i}+1, \overline{i}-1] \right|
	D_{\underline{i}} = \delta_{\underline{i}} \right]}
{\mathbb{P}\left[\left.
	D_{i} \in \Theta_{i} \, \forall\, i \in [\underline{i}+1, \overline{i}-1] \right|
	D_{\underline{i}} = \delta_{\underline{i}} \right]}$
for $\underline{i}+1 \leq \overline{i} - 1$ using the definition of conditional probability. Using \eqref{Definition_P1_tilde}, this can be written as \eqref{eq_P2_tilde}.
\begin{align}
	&\tilde{P}_3(\underline{i},\overline{i}, (\Theta_i)_{\underline{i} \leq i \leq \overline{i}}) =
	\mathbb{P}\left[	D_{i} \in \Theta_{i} \, \forall\, i \in [\underline{i}, \overline{i}]\right]= \label{Definition_P3_tilde}\\
	&\sum\nolimits_{d_0 = \underline{d}}^{\overline{d}}
	\mathbb{P}\left[	D_{0} = d_0 \right]
	\mathbb{P}\left[\left.	D_{i} \in \Theta_{i} \, \forall\, i \in [\underline{i}, \overline{i}] \right| D_{0} = d_0 \right]=\nonumber\\
	& \begin{cases}
		\sum\nolimits_{d_0 \in \Theta_{0}}
		\mu(d_0)
		& 0 = \underline{i} = \overline{i}
		\\
		\sum\limits_{d_0 \in \Theta_{0}}
		\mu(d_0)
		\mathbb{P}\left[\left.	D_{i} \in \Theta_{i} \, \forall\, i \in [1, \overline{i}] \right| D_{0} = d_0 \right]
		& 0 = \underline{i} < \overline{i}
		\\
		\sum\nolimits_{d_0 = \underline{d}}^{\overline{d}}
		\mu(d_0)
		\sum\nolimits_{d_{\underline{i}} \in \Theta_{\underline{i}}}
		\mathbb{P}\left[\left.	D_{\underline{i}} = d_{\underline{i}} \right| D_{0} = d_0 \right]
		& 0 < \underline{i} = \overline{i}
		\\
		\begin{matrix}
			\sum\nolimits_{d_0 = \underline{d}}^{\overline{d}}
			\mu(d_0)
			\sum\nolimits_{d_{\underline{i}} \in \Theta_{\underline{i}}}
			\mathbb{P}\left[\left.	D_{\underline{i}} = d_{\underline{i}} \right| D_{0} = d_0 \right]\\
			\mathbb{P}\left[\left.	D_{i} \in \Theta_{i} \, \forall\, i \in [\underline{i}+1, \overline{i}] \right| D_{\underline{i}} = d_{\underline{i}} \right]
		\end{matrix}
		& 0 < \underline{i} < \overline{i}
	\end{cases} \nonumber
\end{align}
for $\underline{i} \leq \overline{i}$ using the law of total probability and \eqref{eq_mu}. Using \eqref{eq_Phi_n}, this can be written as \eqref{eq_P3_tilde}. $\tilde{P}_4(\underline{i},\overline{i}, \delta_{\overline{i}}, (\Theta_i)_{\underline{i} \leq i \leq \overline{i} - 1}) = 
\mathbb{P}\left[\left.D_{\overline{i}} = \delta_{\overline{i}} \right|
D_{i} \in \Theta_{i}  \forall i \!\in\! [\underline{i}, \overline{i}\!-\!1]\right]
= \frac
{\mathbb{P}\left[D_{\overline{i}} = \delta_{\overline{i}},\,
	D_{i} \in \Theta_{i} \forall i \in [\underline{i}, \overline{i}-1] \right]}
{\mathbb{P}\left[
	D_{i} \in \Theta_{i} \, \forall\, i \in [\underline{i}, \overline{i}-1] \right]}$
for $\underline{i} \leq \overline{i} - 1$ which can be written as \eqref{eq_P4_tilde} using \eqref{Definition_P3_tilde}.

\subsubsection{$\mathit{k \geq K_{sd}}$}
For $\tilde{h}_k = s_k-1$, \eqref{eq_relevant_information} can be written as
$P_k(\delta) = \mathbb{P}\left[\left.D_{k-s_k+1} = \delta \right| 
D_{k-s_k} = d_{k-s_k}\right]
= \Phi(d_{k-s_k}, \delta)$
using \eqref{eq_Phi}. For $\tilde{h}_k < s_k-1$, \eqref{eq_relevant_information} can be written as
$P_k(\delta) = \mathbb{P}\big[D_{k-\tilde{h}_k} = \delta \big| 
D_{k-s_k} = d_{k-s_k},  D_{k+i} \in \Theta_{k,i}\, \forall\, i \in [-s_k+1, -\tilde{h}_k-1]
\big]
= \tilde{P}_2(k-s_k,k-\tilde{h}_k,d_{k-s_k},\delta,(\Theta_{k,i})_{-s_k+1 \leq i \leq -\tilde{h}_k - 1})$.

\subsubsection{$\mathit{k < K_{s}}$}
For $k-\tilde{h}_k \leq \underline{d}$, \eqref{eq_relevant_information} can be written as
\begin{align}
	&P_k(\delta) = \mathbb{P}\big[D_{k-\tilde{h}_k} = \delta \big]= \sum\nolimits_{d_0 = \underline{d}}^{\overline{d}}
	\mu(d_0)
	\Phi_{k-\tilde{h}_k}(d_0,\delta)\label{eq_3}
\end{align}
using \eqref{eq_Phi_n}. For $k-\tilde{h}_k > \underline{d}$, \eqref{eq_relevant_information} can be written as $P_k(\delta)=\mathbb{P}\big[D_{k-\tilde{h}_k} = \delta \big| 
D_{k+i} \in \Theta_{k,i}\, \forall\, i \in [\underline{d}-k, -\tilde{h}_k-1]
\big] 
= \tilde{P}_4(\underline{d},k-\tilde{h}_k, \delta, (\Theta_{k,i})_{\underline{d}-k \leq i \leq -\tilde{h}_k - 1})$.

\subsubsection{$\mathit{K_{s} \leq k < K_{sd}}$}
For $k-\tilde{h}_k \leq \underline{d}$, \eqref{eq_relevant_information} can be written as \eqref{eq_3}. For $k-\tilde{h}_k > \underline{d}$, \eqref{eq_relevant_information} yields $P_k(\delta)=\mathbb{P}\big[D_{k-\tilde{h}_k} = \delta \big| 
D_{k+i} \in \hat{\Theta}_{k,i}\, \forall\, i \in [\underline{d}-k, -\tilde{h}_k-1]
\big]
= \tilde{P}_4\big(\underline{d},k-\tilde{h}_k, \delta, (\hat{\Theta}_{k,i})_{\underline{d}-k \leq i \leq -\tilde{h}_k - 1}\big)$.

\end{document}